\documentclass[sigconf,natbib=false,hidelinks]{acmart}
\setcopyright{rightsretained}

\acmConference[PACT '22]{PACT '22:  International Conference on
Parallel Architectures and Compilation Techniques (PACT)}{October 10--12, 2012}{Chicago,IL}
\acmBooktitle{PACT '22: International Conference on
Parallel Architectures and Compilation Techniques (PACT), October 10--12, 2018, Chicago, IL}
\acmPrice{15.00}
\acmISBN{978-1-4503-XXXX-X/18/06}

\newcommand{\HRule}{\rule{\linewidth}{0.35mm}}

\newcommand{\comm}[1]{}
\newcommand{\lang}[1]{\ensuremath{L(#1)}}
\newcommand{\bisim}[1]{\ensuremath{\approx}}
\newcommand{\semantics}[1]{\ensuremath{\llbracket #1 \rrbracket}}
\newcommand{\expheight}[1]{\ensuremath{\mathit{height}(#1)}}

\newcommand{\chr}[1]{\text{\tt #1}}
\newcommand{\costfn}{\ensuremath{\mathit{cost}}}
\newcommand{\cost}[1]{\ensuremath{\costfn(#1)}}
\newcommand{\rewrite}[1]{\ensuremath{\mathit{rewrite}(#1)}}

\usepackage{lipsum}
\usepackage{wrapfig}

\usepackage{extarrows}
\usepackage{paralist}

\usepackage[linesnumbered,ruled,noend]{algorithm2e}

\SetCommentSty{mycommfont}

\SetKwInput{KwInput}{Input}                
\SetKwInput{KwOutput}{Output}              
\SetKw{New}{new}
\SetKw{Self}{\textit{self}}
\SetKwProg{Fn}{fn}{}{}
\SetKwProg{CFn}{fn}{:}{}
\SetKwProg{Thread}{Thread}{}{}
\SetKwProg{Init}{Init}{}{}
\SetKwProg{Run}{Run}{}{}
\SetKwProg{Method}{Method}{}{}
\SetKwProg{CMethod}{Method}{:}{}

\usepackage{soul}
\usepackage{amsfonts}
\usepackage{mathpartir}
\usepackage{stmaryrd}
\PassOptionsToPackage{usenames,dvipsnames}{xcolor}
\usepackage{listings}          
\lstdefinestyle{ocaml}{
  language=caml,
  morekeywords={struct,module},
  columns=fixed
}
\lstdefinelanguage{JavaScript}{
  keywords={const, typeof, new, true, false, catch, function, return, null, catch, switch, let, var, if, in, while, do, else, case, break},
  keywordstyle=\color{blue},
  ndkeywords={class, export, boolean, throw, implements, import, this},
  ndkeywordstyle=\color{darkgray}\bfseries,
  identifierstyle=\color{black},
  sensitive=false,
  comment=[l]{//},
  morecomment=[s]{/*}{*/},
  commentstyle=\color{purple}\ttfamily,
  stringstyle=\color{red}\ttfamily,
  morestring=[b]',
  morestring=[b]",
  xleftmargin=\parindent
}
\usepackage{listings-rust}
\usepackage{textcomp}
\usepackage{plstx}

\usepackage[colorinlistoftodos,prependcaption,textsize=tiny]{todonotes}

\usepackage{booktabs}   

\usepackage[utf8]{inputenc}
\usepackage[style=numeric, natbib=true, maxbibnames=30, maxcitenames=1, backend=biber, safeinputenc]{biblatex} 
\renewbibmacro{in:}{ \printtext{\bibstring{in}} }

\usepackage{tabularx}
\usepackage{mathpartir}
\usepackage{stmaryrd}
\usepackage{listings}          
\usepackage{plstx}

\usepackage{tikz}
\usetikzlibrary{automata}
\usepackage{pgfplots}
\usetikzlibrary{external}
\tikzsetexternalprefix{tikzplots/}
\tikzsetfigurename{main-figure}
\usepackage{filecontents}

\newcommand{\pages}[1]{}

\usepackage{tikzit}

\tikzstyle{new style 0}=[fill=white, draw=black, minimum size=4pt, shape=circle, inner sep=4pt]
\tikzstyle{tree_node}=[fill=black, draw=black, shape=circle, minimum size=4pt, inner sep=0pt]
\tikzstyle{edge_label}=[fill=none, draw=none, shape=circle, inner sep=0pt, minimum size=10pt]
\tikzstyle{tree_node}=[fill=white, draw=black, shape=circle, inner sep=1pt, minimum size=3pt]
\tikzstyle{rect}=[fill=none, draw=black, shape=rectangle, dotted]
\tikzstyle{solid_black}=[fill=black, draw=black, shape=circle, minimum size=3pt, inner sep=0pt]
\tikzstyle{big_dashed}=[fill=none, draw={rgb,255: red,128; green,128; blue,128}, shape=circle, dotted, minimum size=14pt, thick]

\tikzstyle{new edge style 0}=[draw=black, label=Hello, ->]
\tikzstyle{new edge style 1}=[->, draw=black, dashed]
\tikzstyle{accepted}=[draw=blue, ->, line width=1.5pt]
\tikzstyle{accepted_dash}=[draw=blue, ->, line width=1.5pt, dashed]
\tikzstyle{black solid}=[fill=none, ->, thick]
\tikzstyle{reg_black}=[->]
\tikzstyle{red_del}=[draw=red, dashed, ->]
\tikzstyle{gray_solid}=[draw={rgb,255: red,128; green,128; blue,128}, ->]
\tikzstyle{blue_solid}=[draw=blue, ->, thick]

\newtheorem{innercustomthm}{Theorem}
\newenvironment{customthm}[2]
  {\renewcommand\theinnercustomthm{\ref{#1}}\innercustomthm[#2]}
  {\endinnercustomthm}

\usepackage{balance}

\usepackage[symbol]{footmisc}

\begin{document}


\title{Optimizing Regular Expressions via Rewrite-Guided Synthesis}

\author[Jedidiah McClurg]{Jedidiah McClurg}
\affiliation{\institution{Colorado State University}\country{USA}}
\author[Miles Claver]{Miles Claver}
\authornote{Equal contribution}
\affiliation{\institution{Colorado School of Mines}\country{USA}}
\author[Jackson Garner]{Jackson Garner\footnotemark[1]}
\affiliation{\institution{Colorado School of Mines}\country{USA}}
\author[Jake Vossen]{Jake Vossen\footnotemark[1]}
\affiliation{\institution{Colorado School of Mines}\country{USA}}
\author[Jordan Schmerge]{Jordan Schmerge}
\affiliation{\institution{Colorado School of Mines}\country{USA}}
\author[Mehmet E. Belviranli]{Mehmet E. Belviranli}
\affiliation{\institution{Colorado School of Mines}\country{USA}}



%


\begin{abstract}
Regular expressions are pervasive in modern systems.
Many real-world regular expressions are inefficient, sometimes to the extent that they are vulnerable to complexity-based
attacks, and while much research has focused on \textit{detecting} inefficient regular expressions or
\textit{accelerating} regular expression matching at the hardware level,
we investigate automatically \textit{transforming} regular expressions to remove inefficiencies.
We reduce this problem to general \textit{expression optimization}, an important task necessary in a variety of
domains even beyond compilers, e.g., digital logic design, etc.
Syntax-guided synthesis (SyGuS) with a cost function can be used for this purpose,
but ordered enumeration through a large space of candidate expressions can be prohibitively expensive.
Equality saturation is an alternative approach which allows efficient construction
and maintenance of expression equivalence classes generated by rewrite rules,
but the procedure may not reach saturation, meaning global minimality cannot be confirmed.
We present a new approach called rewrite-guided synthesis (ReGiS), in which a unique interplay between
SyGuS and equality saturation-based rewriting helps to overcome these problems,
resulting in an efficient, scalable framework for expression optimization.
%

\end{abstract}

\maketitle


\section{Introduction \pages{(3 pages)}}
\label{sec:intro}

Because regular expressions and their associated operations (matching, etc.) play such a pivotal role in modern systems,
there has been much interest in developing hardware acceleration for regular expressions
\cite{DBLP:conf/IEEEpact/YangP11,DBLP:conf/micro/LunterenHHBSA12,DBLP:conf/IEEEpact/CameronSSH0HL14,DBLP:conf/micro/GogteKCDW16,DBLP:journals/tecs/ParraviciniCSPS21,DBLP:conf/pldi/Mamouras22}.
Our work investigates a complementary approach, namely optimizing regular expressions at the software level.
Because there are other popular formalisms that share similar properties to regular expressions (e.g., Boolean algebra), we
frame the problem in terms of general expression optimization, enabling straightforward extensions in other domains.

Expression optimization is a type of \textit{program synthesis} problem---we must automatically construct a program (expression) that
satisfies some specification (e.g., minimal cost, and equality to the input expression).
In the mid-Eighties, \citet{DBLP:journals/computer/Brooks87} famously identified several technological areas unlikely
to result in a ``silver bullet'' in terms of increased programmer productivity and software quality, and
program synthesis appeared in the list.
Since then, significant strides have been made in some of these areas, perhaps most
notably, data-centric advances in machine learning
which have enabled software to perform a variety of complex tasks, including winning
chess matches against professionals, driving cars, and landing rockets.
Overall, progress in the area of program synthesis
has seen more moderate gains.
One notable approach is \textit{syntax-guided synthesis} (SyGuS) \cite{DBLP:series/natosec/AlurBDF0JKMMRSSSSTU15},
which has leveraged domain-specific languages (DSLs) and exploited fast solvers
(e.g., SAT and SMT \cite{DBLP:journals/jacm/NieuwenhuisOT06}) to produce synthesizers usable in areas such as
distributed systems \cite{DBLP:conf/pldi/UdupaRDMMA13}, robotics \cite{DBLP:journals/corr/ChasinsN16}, biochemical modeling \cite{DBLP:conf/cav/CardelliCFKLPW17}, networking
\cite{DBLP:conf/cav/McClurgHC17}, and many more. Conceptually, SyGuS 
performs a search over the space of all program expressions, checking at each step
if the expression satisfies the specification. Although various techniques have
been devised to make this search more efficient, many of the ``big ideas'' that
have allowed advancement elsewhere (big data, novel hardware
processing units, massive parallelization) have proven more difficult to
utilize in this type of syntax-guided search.

\subsection{Problem Description: Expression Optimization}
In this paper, we develop a new optimal synthesis framework called Rewrite-Guided
Synthesis (ReGiS) which extends SyGuS, making it more flexible and
amenable to parallelization.
Our goal is to take an initially-correct
\textit{source expression}, and transform it into a \textit{better} equivalent
expression. The user can provide the expression language,
an optional set of semantics-preserving rewrite rules, a cost metric for
expressions, and a source expression, and the synthesizer outputs an equivalent
expression that is minimal with respect to the cost metric.

\subsection{Existing Approaches}
Several existing approaches can be used for expression optimization.
\textit{Optimal Synthesis} \cite{DBLP:conf/popl/BornholtTGC16,DBLP:conf/cav/CardelliCFKLPW17} uses a cost metric and techniques such as counterexample-guided enumeration to search for an optimal program satisfying a specification.
\textit{Rewriting} \cite{DBLP:journals/corr/abs-1012-1802,DBLP:journals/pacmpl/WillseyNWFTP21} uses syntactic transformations and efficient data structures to produce equivalent expressions with differing structure.
\textit{Superoptimization} \cite{DBLP:conf/asplos/Schkufza0A13,DBLP:conf/asplos/PhothilimthanaT16} transforms small snippets of code into equivalent and higher-performing snippets, using enumerative or rewriting-based methods.
Section \ref{sec:rel_work} gives more detail about these approaches.
In contrast to these, ReGiS uses a unique combination of enumeration and rewriting, resulting in a more flexible and efficient technique.

\subsection{ReGiS Novelties}
ReGiS targets three core improvements over previous approaches.


\paragraph{(1) Combining enumeration (using semantic correctness/equality) and syntactic rewriting.}
Enumerative search---symbolically or explicitly iterating through program expressions in increasing order
with respect to cost while checking semantic correctness/equality---is often not efficient when the goal is to
\textit{optimize} a given input expression in some way, i.e., transform it into
an equivalent expression with lower cost.
Specifically, since the input expression is already \textit{correct},
it may be counterproductive to ``start from scratch'' when building an
equivalent expression.
In many domains, it is possible to find semantics-preserving
transformations \cite{DBLP:conf/cav/CernyHRRT13} which allow \textit{rewriting} an expression to obtain
lower cost, with respect to a metric like expression size or time complexity.
In some cases, these transformation rules have useful properties, e.g., soundness and completeness in the case of Kleene algebra
for regular expressions
\cite{DBLP:conf/lics/Kozen91}, but other times, this is not the case.
Thus, just as purely enumerative synthesis has drawbacks, so too does a purely
rewrite-based approach, since it requires careful design of the rewrite rules.
Additionally, the optimal target expression may have a large distance from the
source expression with respect to the rules, and rewriting-based approaches
can become rapidly overwhelmed as the search depth increases.
For these reasons, we show how to combine enumeration with rewriting, 
allowing exploration of expressions which are \textit{locally close} (syntactically related) to seen expressions, as well as expressions which are \textit{globally small} (having overall lowest cost).

\paragraph{(2) Using parallelizable bi-directional search.}
Rather than simply starting from the
source expression, and trying to discover a chain of equivalences to a specific
target expression, we additionally try to construct these chains
\textit{backward} toward the source from several candidate targets
simultaneously.

\paragraph{(3) Enabling customizable expression languages and semantics.}
Our approach is cleanly parameterized over a user-specifiable expression language. While we focus on the domain of regular expressions, the approach would be equally applicable in other domains such as Boolean logic, process algebras, etc.

\begin{figure}[b]
\centering
{\begin{tikzpicture}[scale=0.72,every node/.style={scale=1.0}]
	\begin{pgfonlayer}{nodelayer}
		\node [style={big_dashed}, label={right:\normalsize{0}}] (0) at (3, 2) {};
		\node [style={big_dashed}, label={right:\normalsize{1}}] (1) at (3, 4) {};
		\node [style={big_dashed}, label={[align=left]east:Ordered\\target\\E-classes}] (2) at (3, 6) {};
		\node [style={big_dashed}, label={right:\normalsize{3}}] (3) at (3, 8) {};
		\node [style={big_dashed}, label=Source E-class] (4) at (-10, 5) {\phantom{ Source }};
		\node [style={solid_black}] (5) at (3, 6) {};
		\node [style={solid_black}] (6) at (3, 2) {};
		\node [style={solid_black}] (7) at (3, 4) {};
		\node [style={solid_black}] (8) at (3, 8) {};
		\node [style={solid_black}, label={above:\normalsize{Source}}] (9) at (-10, 5) {};
	\end{pgfonlayer}
	\begin{pgfonlayer}{edgelayer}
		\draw [>=stealth, style={red_del}, bend right=15, looseness=1.00] (4) to node[auto, style={edge_label}]{\color{red}\st{$1$} $\infty$} (0);
		\draw [>=stealth, style={blue_solid}, bend right=7, looseness=0.75] (4) to node[auto, style={edge_label}]{\color{blue}\st{$1$} $2$} (1);
		\draw [style={gray_solid}] (2) to (1);
		\draw [style={gray_solid}] (3) to (2);
		\draw [style={gray_solid}] (1) to (0);
		\draw [>=stealth, style={reg_black}, bend left=7, looseness=1.00] (4) to node[auto, style={edge_label}]{$1$} (2);
		\draw [>=stealth, style={reg_black}, bend left=15, looseness=1.00] (4) to node[auto, style={edge_label}]{$1$} (3);
	\end{pgfonlayer}
\end{tikzpicture}}
    \caption{Overlay graph: edge labels encode \textit{expense estimates}; dashed/red edge shows an \textit{inequality} discovered by a Unifier, which causes edge deletion; and thick/blue edge shows a Unifier \textit{timeout}, which increases estimate.}
\label{fig:overview}
\end{figure}
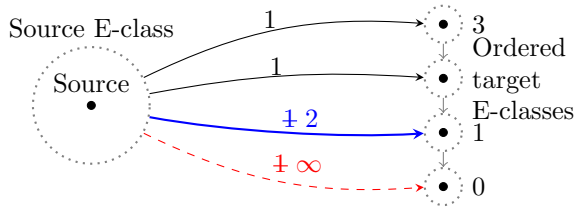

\subsection{ReGiS Approach Overview}

ReGiS consists of three components:
\textit{Enumerator}, \textit{Updater}, and \textit{Unifier}.
The Enumerator iterates through candidate target expressions in increasing
order of cost, adding each new candidate to the Updater.
When the Updater receives a new target expression, it is added to an efficient 
\textit{E-graph} data structure \cite{DBLP:journals/corr/abs-1012-1802,DBLP:journals/pacmpl/WillseyNWFTP21}, allowing \textit{all} known rewrite rules to be applied to the
expression, which enables compact maintenance of the \textit{equivalence 
classes
(E-classes)} for the source and candidate target expressions, modulo the
known syntactic rewrite rules.
The Updater also maintains an
\textit{overlay graph}
(Figure \ref{fig:overview}),
with E-classes as
nodes, and edge labels representing the (initially unit) estimated
expense of semantic equality checks between classes.
In parallel with these processes, Unifiers systematically attempt to
merge E-classes: each Unifier selects a low-cost overlay graph edge, chooses expressions from the two corresponding E-classes, and performs a semantic equality check.
If its equality check \textit{succeeds}, it tells the Updater to union the two
E-classes, and
can potentially provide the Updater with a new rewrite rule(s).  
If its equality check \textit{fails}, it removes the associated edge.
If its equality check \textit{times out}, it increases its edge's expense estimate.
Eventually, a target E-class that is minimal with respect to cost will
be unioned with the source E-class, %
allowing ReGiS to terminate and report the global minimum.
The current lowest-cost result is available as the
minimum-cost expression in the source E-class.
In contrast to approaches  that extend SyGuS by parallelizing 
enumeration steps \cite{DBLP:conf/cav/JeonQSF15},
our approach does the syntactic (rewriting-based) and semantic (equality-based)
parts of the search in parallel.

\subsection{Paper Organization}

This paper is organized as follows:
$\S$\ref{sec:motiv} demonstrates why
superlinear regular expressions are problematic, and shows how ReGiS can be used to address this; %
$\S$\ref{sec:regis} formalizes our approach, and presents correctness results;
$\S$\ref{sec:redos} shows the details of using ReGiS for regular expression optimization;
$\S$\ref{sec:eval} describes our prototype implementation, and provides experimental results;
$\S$\ref{sec:discussion} and $\S$\ref{sec:rel_work} describe future work and related work; and $\S$\ref{sec:conclusion} concludes.

\section{Motivating Example: Optimizing Superlinear Regular Expressions \pages{(4 pages)}}
\label{sec:motiv}

We demonstrate the utility of our framework by examining problematic behavior of
superlinear regular expressions.
\textit{Catastrophic backtracking} behavior can be triggered
by crafting input strings to
target inefficiencies in the regular expression.
As an example, consider the regular expression
$R_1 R_2 = \chr{a}^{\ast}\chr{a}^{\ast}$.
If we try
to match the entire input string $c_1 c_2 c_3 \ldots c_n c_{n+1} = \chr{a}\chr{a}\chr{a}\ldots \chr{a}\chr{b}$
using this regular expression, we might first greedily capture $c_1\ldots c_n$
using $R_1$, only to realize that there is no way to match the trailing
\chr{b}.
We would then need to backtrack and accept $n-1$ leading
\chr{a} characters with $R_1$, and let $R_2$
match the final \chr{a}, which would similarly fail due to the trailing
\chr{b} in the input.
This would continue, with $R_1$ accepting $c_1\ldots c_k$,
and $R_2$ accepting $c_{k+1}\ldots c_n$, until all $k$ have been tried,
resulting in quadratic runtime.

One way to avoid this issue is to use non-backtracking 
algorithms.
For example, we could convert the regular expression to a
\textit{nondeterministic finite automaton} (NFA) using Thompson's construction \cite{Thompson}, and then determinize the NFA,
but this can result in exponential explosion of the automaton size, so this
approach is not typically used in practice.
\citet{Thompson} also presented an automaton simulation algorithm which can match a string against an NFA in polynomial time.
Unfortunately, many real-world regular expression engines have chosen to instead
rely on backtracking algorithms,
due to complex extensions to the regular
expression language (backreferences, etc.).
\textit{Perl-Compatible Regular Expressions} (PCRE) is one such implementation
\cite{DBLP:journals/corr/BerglundDM14}.

These superlinear
regular expressions appear with concerning frequency
in real-world systems \cite{DBLP:conf/uss/StaicuP18,Davis1}, and
real attacks 
have been observed.
%
As an orthogonal approach to ours, \textit{static analysis} has been used to
\textit{detect} exponential
regular expressions \cite{DBLP:journals/corr/RathnayakeT14}.
Note that focusing on
\textit{exponential} regular expressions
is insufficient---although the maximum number of operations for backtracking regular expression algorithms is bounded by $2^{\Theta(n)}$ \cite{DBLP:journals/iandc/HromkovicSKKS02,palioudakis2015quantifying}, \textit{polynomial} complexity
can also be problematic \cite{Weideman}.

Several approaches have been identified for dealing with superlinear
regular expressions \cite{Davis1}, 
the most promising of which seems to be transforming the
expression
into an equivalent but less complex one. To our knowledge, however, this has
not been solved in a comprehensive way.
In this section, we will examine the problem of optimizing superlinear regular expressions in greater detail, and show how the various components of our approach work together to tackle this problem.
Consider Figure \ref{fig:superlinear}, which shows the performance of the standard PCRE matching algorithm for the
regular expressions $\chr{a}^{\ast}$ and
$\chr{a}^{\ast\ast}$.
These regular expressions are \textit{semantically equivalent}, i.e., they
recognize the same language, but their
differing syntactic structures cause
drastically different performance when matching the
previously-described input string
$\chr{a}\chr{a}\ldots\chr{a}\chr{b}$.
Regular expression $\chr{a}^{\ast}$ has
linear performance $\Theta(n)$, while $\chr{a}^{\ast\ast}$ has exponential performance $\Theta(2^n)$, and each additional added star increases the base of the exponent.
Intuitively, at each step, $\chr{a}^{\ast}$ has only two options: accept
a \textit{single} \chr{a} character or fail on the
trailing \chr{b} character,
but $\chr{a}^{\ast\ast}$ can accept an arbitrarily-long \textit{sequence} of \chr{a} characters at each step,
forcing the algorithm to try all possible
combinations of sequence lengths before failing.
In Sec. \ref{sec:simp}, we cover this example in more detail,
and introduce a cost metric that characterizes such backtracking behavior.

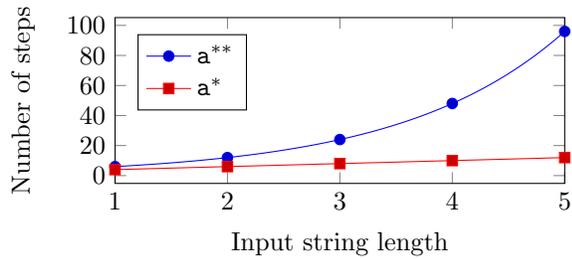
\begin{figure}[tb]
    \centering
    \begin{tikzpicture}[scale=0.99]
        \begin{axis}[width=3in,height=1.5in,ytick={0,20,40,60,80,100},xlabel={Input string length},ylabel={Number of steps},xmin=1,xmax=5,domain=0:3145728, restrict y to domain=0:100, legend style={at={(0.05,0.900)},anchor=north west},legend cell align={left}]
    \addplot[raw gnuplot,smooth,blue,forget plot] gnuplot {
    f(x)=a + b*(2^(c*x));
         fit f(x) 'figures/superlinear.csv' using 1:2 via a,b,c;
         plot [x=0:10] f(x);
       };
    \legend{$\chr{a}^{\ast\ast}$,$\chr{a}^{\ast}$}       
    \addplot table [x=len, y=astarstar, col sep=tab, only marks] {figures/superlinear.csv};
    \addplot table [x=len, y=a, col sep=tab] {figures/superlinear.csv};
\end{axis}
\end{tikzpicture}
\caption{Matching w/ semantically-equivalent expressions (input $\chr{a}\chr{a}\ldots\chr{a}\chr{b}$).}
\label{fig:superlinear}
\end{figure}


\subsection{Limitations of Basic Rewriting}
\label{subsec:limitations}
One basic optimization approach
is to perform rewriting using the
well-known Kleene algebra axioms \cite{DBLP:conf/lics/Kozen91}, at each step checking
whether we have found an expression that has lower cost
according to our metric.
For example, given $\chr{a}+\chr{a}+\chr{a}$ (where $+$ denotes alternation), we can
use the idempotence rule $x{+}x\xleftrightarrow{}x$
to perform the rewrites
$\chr{a}{+}\chr{a}{+}\chr{a}\,\xrightarrow{}\,\chr{a}{+}\chr{a}\,\xrightarrow{}\,\chr{a}$, and we will have reached
an equivalent regular expression with lower cost.

This basic approach scales poorly---in
general, we would need to perform a rewrite-based
search, iterating through the various rewrite rules.
The search can ``loop'',
e.g., rewriting an expression into progressively larger
expressions.
Note that we cannot restrict rewrites to only
\textit{shrink} expression cost, because in some cases,
\textit{global} minimization necessitates local
monotonic (or even increasing) rewrites during the search.
As an example, optimizing $1+\chr{a}^*$ (where $1$ denotes the empty string) requires a rewrite which initially increases cost. Specifically, using arrow angle to indicate change in cost due to a rewrite, we have $1+\chr{a}^* {\nearrow 1+1+\chr{a}\chr{a}^*} \searrow 1+\chr{a}\chr{a}^* \searrow \chr{a}^*$.
Regular expression optimizers based on
this type of rewrite-based search often timeout before making any progress.
For regular expressions such as
$\chr{a}+\chr{b}+\chr{c}+\chr{d}+\chr{e}+\chr{d}+\chr{c}+\chr{b}+\chr{a}$ (which is
clearly reducible to $\chr{a}+\chr{b}+\chr{c}+\chr{d}+\chr{e}$),
the search would need to conceptually ``sort'' the
characters using commutativity of alternation,
and then use idempotence, requiring a huge amount of search.

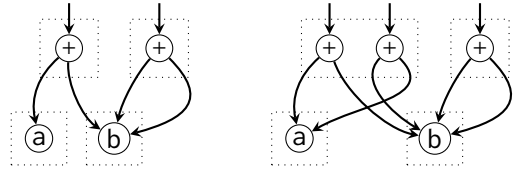
\begin{figure}
\centering
\begin{tikzpicture}[scale=0.8]
	\begin{pgfonlayer}{nodelayer}
		\node [style=rect, minimum width=22pt, minimum height=22pt] (34) at (-3, 4) {};
		\node [style=rect, minimum width=22pt, minimum height=22pt] (35) at (0, 4) {};
		\node [style={tree_node}] (25) at (-3, 4) {\normalsize{$+$}};
		\node [style={tree_node}] (27) at (0, 4) {\normalsize{$+$}};
		\node [style={tree_node}] (28) at (-4, 1) {\Large$\chr{a}$};
		\node [style={tree_node}] (29) at (-1.5, 1) {\Large$\chr{b}$};
		\node [style=none] (31) at (-3, 5.5) {};
		\node [style=none] (32) at (-1, 5.5) {};
		\node [style=none] (33) at (0, 5.5) {};
		\node [style=rect, minimum width=21pt, minimum height=20pt] (36) at (-4, 1) {};
		\node [style=rect, minimum width=21pt, minimum height=20pt] (37) at (-1.5, 1) {};
	\end{pgfonlayer}
	\begin{pgfonlayer}{edgelayer}
		\draw [>=stealth, style=black solid] (31.center) to (25);
		\draw [>=stealth, style=black solid] (33.center) to (27);
		\draw [>=stealth, style=black solid, bend right] (25) to (28);
		\draw [>=stealth, style=black solid, bend right, looseness=0.75] (25) to (29);
		\draw [>=stealth, style=black solid, bend right=15, looseness=0.75] (27) to (29);
		\draw [>=stealth, style=black solid, in=15, out=-45, looseness=1.75] (27) to (29);
	\end{pgfonlayer}
\end{tikzpicture}\hspace{2em}%
\begin{tikzpicture}[scale=0.8]
	\begin{pgfonlayer}{nodelayer}
		\node [style=rect, minimum width=44pt, minimum height=22pt] (34) at (-2, 4) {};
		\node [style=rect, minimum width=22pt, minimum height=22pt] (35) at (2, 4) {};
		\node [style={tree_node}] (25) at (-3, 4) {\normalsize{$+$}};
		\node [style={tree_node}] (27) at (2, 4) {\normalsize{$+$}};
		\node [style={tree_node}] (28) at (-4, 1) {\Large$\chr{a}$};
		\node [style={tree_node}] (29) at (0.5, 1) {\Large$\chr{b}$};
		\node [style={tree_node}] (30) at (-1, 4) {\normalsize{$+$}};
		\node [style=none] (31) at (-3, 5.5) {};
		\node [style=none] (32) at (-1, 5.5) {};
		\node [style=none] (33) at (2, 5.5) {};
		\node [style=rect, minimum width=21pt, minimum height=20pt] (36) at (-4, 1) {};
		\node [style=rect, minimum width=21pt, minimum height=20pt] (37) at (0.5, 1) {};
	\end{pgfonlayer}
	\begin{pgfonlayer}{edgelayer}
		\draw [>=stealth, style=black solid] (31.center) to (25);
		\draw [>=stealth, style=black solid] (32.center) to (30);
		\draw [>=stealth, style=black solid] (33.center) to (27);
		\draw [>=stealth, style=black solid, bend right] (25) to (28);
		\draw [>=stealth, style=black solid, bend right, looseness=0.75] (25) to (29);
		\draw [>=stealth, style=black solid, in=30, out=-45, looseness=1.25] (30) to (28);
		\draw [>=stealth, style=black solid, in=150, out=-135] (30) to (29);
		\draw [>=stealth, style=black solid, bend right=15, looseness=0.75] (27) to (29);
		\draw [>=stealth, style=black solid, in=15, out=-45, looseness=1.75] (27) to (29);
	\end{pgfonlayer}
\end{tikzpicture}
\vspace{4pt}
\caption{(a) E-graph initially built from $\chr{a}+\chr{b}$ and $\chr{b}+\chr{b}$, and (b) after equality saturation using rewrite rule $x + y \xleftrightarrow{} y + x$.}
\label{fig:egraph}
\end{figure}

\subsection{Limitations of E-Graph-based Rewriting}
\textit{Equality saturation} is a technique for
efficiently implementing a rewriting-based task
such as the one previously described.
This approach uses a data structure called
an E-graph to compactly store one or more initial expressions,
along with expressions derivable from these via
a set of rewrite rules.
Figure \ref{fig:egraph}(a) shows an example, namely the
E-graph containing regular expressions
$\chr{a} + \chr{b}$ and $\chr{b}+\chr{b}$. Equality saturation can
apply the \textit{commutativity} rewrite rule, which adds
the expression $\chr{b}+\chr{a}$ to 
the E-graph, resulting in Figure \ref{fig:egraph}(b).
Note that each subexpression \chr{a} and \chr{b}
is stored only once---the E-graph maintains this
type of expression sharing to keep the size
compact.
The dotted boxes in the figure represent
E-classes---equivalence classes with respect
to the rewrite rules.
Expressions $\chr{a} + \chr{b}$ and $\chr{b} + \chr{a}$
are in the same E-class, since they
are equivalent with respect to the rewrite rule, but
$\chr{b}+\chr{b}$ is in a separate E-class.

With an
E-graph-based rewriting approach,
the straightforward way to
implement regular expression optimization is to first add the
source expression to the E-graph,
run equality saturation using all
of the Kleene algebra axioms as
rewrite rules, and iterate
over the source regular expression's
E-class to find the minimal
equivalent expression with respect
to the cost metric.
There are two key problems with this.
\begin{inparaenum}[(1)]
\item Although \textit{cyclic} edges in the E-graph can sometimes be used to encode infinite sets, in general, equality saturation may not have enough time or resources to
fully \textit{saturate} the E-graph in cases where there are infinitely many equivalent expressions
with respect to the Kleene algebra axioms
	(e.g., $\chr{a} = \chr{a}+\chr{a}=\chr{a}+\chr{a}+\chr{a}=\cdots$),
meaning the procedure may need to time out.
\item ReGiS is designed to be \textit{general}, and in some cases, we may have a more limited set of rewrite rules---in particular, we may not have a \textit{completeness} result, meaning that for some semantically equivalent expressions, it may not be possible to show their equivalence using the syntactic rewrite rules alone.
\end{inparaenum}

As an example, consider optimizing $(1+\chr{a}^{\ast}\chr{a})^{\ast\ast}$, using only two rewrite rules: $1+x x^{\ast}\xleftrightarrow{1} x^{\ast}$, $x^{\ast\ast} \xleftrightarrow{2} x^{\ast}$.
What we would need is a chain of rewrites:
$$(1+\chr{a}^{\ast}\chr{a})^{\ast\ast} \,\xrightarrow{2}\, (1+\chr{a}^{\ast}\chr{a})^{\ast} {\color{black}\,\xrightarrow{?}\,} (1+\chr{a}\chr{a}^{\ast})^{\ast} \,\xrightarrow{1}\, \chr{a}^{\ast\ast} \,\xrightarrow{2}\, \chr{a}^{\ast}$$

\noindent
Here, it is not possible to build this chain of equalities using the available syntactic rewrite rules, so we would need a \textit{semantic} equality check to ``bridge the gap'' between 
$(1+\chr{a}^{\ast}\chr{a})^{\ast}$ and  $(1+\chr{a}\chr{a}^{\ast})^{\ast}$.

\subsection{Enumerative Bidirectional Rewriting}

This is the basic idea of
our enumerative bidirectional
rewriting approach.
We use a SyGuS-based strategy to enumerate candidate target regexes by increasing cost,
and adding them to the E-graph.
Equality saturation applies rewrites to the source and all targets simultaneously.
For any target whose E-class intersects the source's E-class, the E-graph will \textit{union} these E-classes.
We iterate through E-classes which are currently disjoint but potentially equal,
and try to equate these using a semantic equality check (NFA bisimilarity).
In this example, a successful equality check $\chr{a}^{\ast}\chr{a} = \chr{a}\chr{a}^{\ast}$ could result in
a new rewrite rule $a^{\ast}a \leftrightarrow a a^{\ast}$,
allowing equality saturation to ``bridge the gap'' indicated by ``${\color{black} \xrightarrow{?}}$''.

\section{ReGiS: Rewrite-Guided Synthesis \pages{(6 pages)}}
\label{sec:regis}

In this section, we formalize our rewrite-guided synthesis approach, and describe key properties of the algorithm.
In Section \ref{sec:simp}, we show in detail how our approach can be used to tackle the real-world problem of optimizing superlinear regular expressions.


\subsection{Expression Optimization}
\label{subsec:expr_opt}

We first specify the problem statement.
Let $G$ be a grammar, and let $\mathcal{E} = L(G)$ be
$G$'s \textit{language}, i.e., the set of expressions
that can be built from $G$.
Let $\mathit{height}: \mathcal{E} \rightarrow \mathbb{N}$ denote height of an
expression's tree.
Let $\mathit{subexprs}: \mathcal{E} \rightarrow \mathcal{P}(\mathcal{E})$ denote
subexpressions.

Let $\costfn: \mathcal{E} \rightarrow \mathbb{R}$ be a \textit{cost function}
that assigns a numeric cost to each expression.
Let $\semantics{\cdot} : \mathcal{E} \rightarrow D$ denote the \textit{semantics}
of the expression language, i.e., a function that maps expressions
to objects of some domain $D$,
and let $\bisim{}\,: (D \times D) \rightarrow \mathbb{B}$ be a
\textit{semantic equality} function for comparing objects in that domain.
Let $hl: \mathcal{E} \rightarrow \mathcal{P}(\mathcal{E})$ denote equivalent expressions of
equal or lesser height, i.e., 
$hl(e) = \{e' \in \mathcal{E} \,|\, \expheight{e'} \leq \expheight{e} \text{ and }\allowbreak \semantics{e'}\bisim{}\semantics{e}\}$.
Given a grammar $G$, we define a \textit{pattern} to be an expression
initially built from $G$, where zero or more subexpressions have been replaced with
\textit{variables} from a set $V$.
{Intuitively, variables serve as placeholders for arbitrary subexpressions built from $G$.}
If $p$ is a pattern, and $m: V \rightarrow \mathcal{E}$ is a mapping, we use $p[m]$ to denote
the expression formed by applying $m$ to each variable in $p$.
{Note that if $p$ contains no variables, $p[m]=p$ for all $m$, and if $\emptyset$ denotes the empty map, $p[\emptyset]=p$ for all $p$.}
We define a \textit{rewrite rule} to be an object of the
form $p_1 \rightarrow p_2$, where $p_1,p_2$ are patterns, and
a \textit{bidirectional} rewrite rule to be of the form
$p_1 \leftrightarrow p_2$.
We say rewrite rule $p_1 \rightarrow p_2$ \textit{matches}
$e$ if and only if there is a mapping $m: V \rightarrow \mathcal{E}$
such that $p_1[m] = e$, and in this case, we say
that $\rewrite{e, p_1 \rightarrow p_2} = \{p_2[m]\}$.
If $e$ does not match $p_1$, then $\rewrite{e, p_1 \rightarrow p_2} = \emptyset$.
If $E$ is a set of expressions and $W$ is a set of rewrite rules,
$\rewrite{E, W}$ signifies
$\bigcup_{e \in E, w \in W} \rewrite{e, w}$.

Let $W$ be a \textit{sound} set of rewrite rules,
i.e., for any $w \in W$,
if $e' \in \rewrite{e, w}$, then $\semantics{e'}\bisim{}\semantics{e}$.
An \textit{optimization instance} is a tuple 
$(e,\costfn,W,\semantics{\cdot},\bisim{})$ where $e \in \mathcal{E}$,
and the \textit{optimization} problem consists of finding
a \textit{minimal} equivalent expression, i.e., an $e' \in \mathcal{E}$
such that $\semantics{e'}\bisim{}\semantics{e}$
and for any $e''$ where $\semantics{e''}\bisim{}\semantics{e}$,
we must have $\cost{e'} \leq \cost{e''}$.

\subsection{E-Graphs}
\label{subsec:lang}

Given an optimization instance, we 
encode the expression language $\mathcal{E}$ using 
the equality saturation framework Egg \cite{DBLP:journals/pacmpl/WillseyNWFTP21},
which accepts a straightforward s-expression-based formulation of the grammar.
Although Egg contains significant machinery to ensure that E-graphs are maintained compactly,
for our formalization purposes, we consider an \textit{E-graph} 
to be a mapping of the form $E: \mathcal{E} \rightarrow (\mathbb{N} \times \mathcal{E})$,
i.e., each contained expression $e$ within E-graph $E$ is associated with a numeric E-class identifier $E_{id}(e)$
and the minimum-cost expression $E_{min}(e)$ within that E-class.
We use $\mathit{class}(E, e)$ to denote the set of all expressions contained in the same E-class
as $e$.

\begin{figure*}[tb]
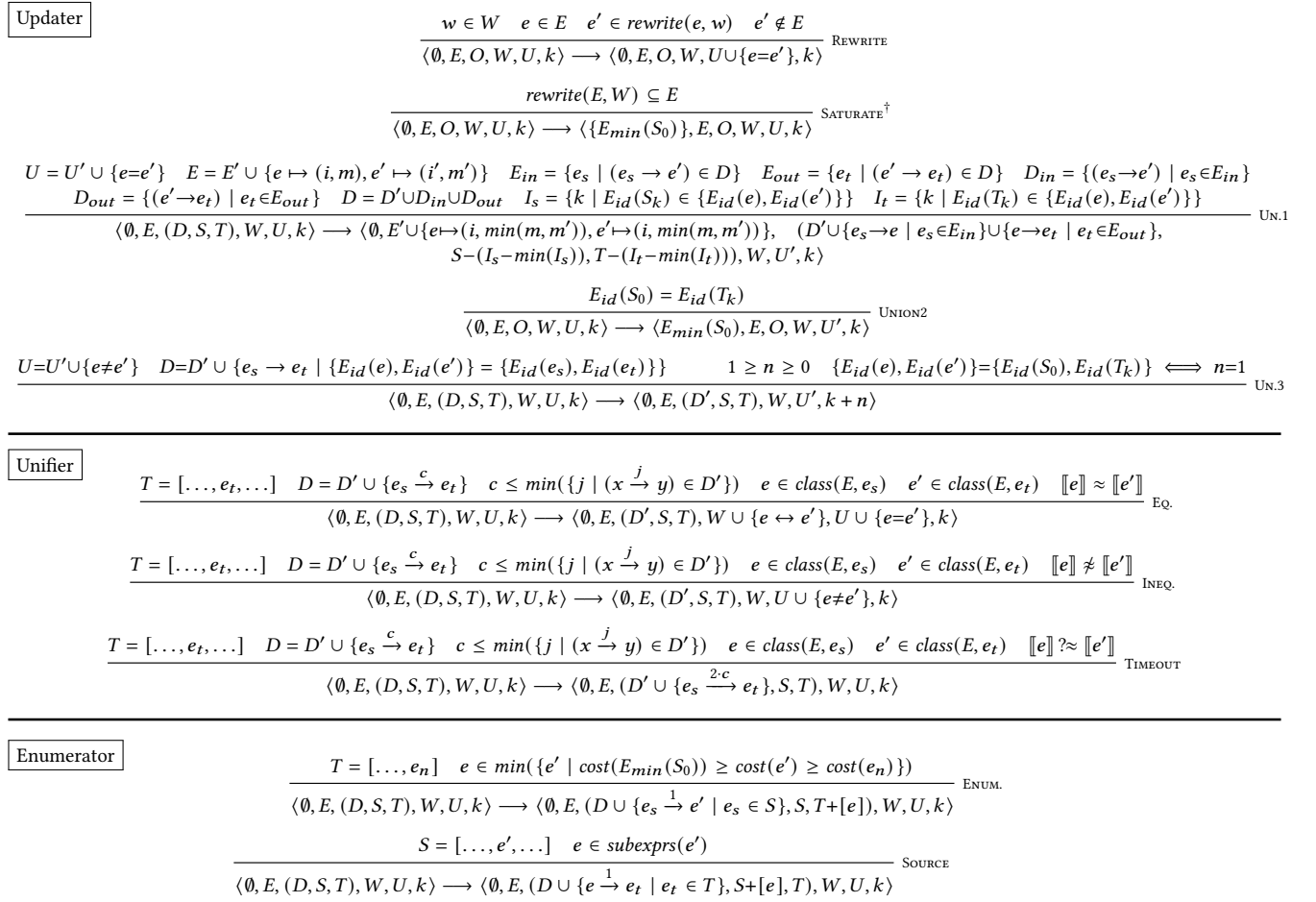

	\scalebox{0.99}
    {\small
\begin{mathpar}
\mprset{sep=1.00em}
	\fbox{Updater} \hfill \vspace{-2.5em} \\

	\inferrule*[Right={\scriptsize Rewrite},leftskip=16pt]{ 
w \in W \\
e \in E \\
e' \in \mathit{rewrite}(e, w) \\
e' \not\in E \\
}{ 
	\langle \emptyset, E, O, W, U, k \rangle \longrightarrow \langle \emptyset, E, O, W, U {\cup} \{e{=}e'\}, k \rangle
}
\vspace{-0.5em} \\

\inferrule*[Right={\scriptsize Saturate$^\dagger$},leftskip=32pt]{ 
\mathit{rewrite}(E, W) \subseteq E\\
}{ 
	\langle \emptyset, E, O, W, U, k \rangle \longrightarrow \langle \{E_{min}(S_0)\}, E, O, W, U, k \rangle
}
\vspace{-0.25em} \\

\inferrule*[Right={\scriptsize Un.1},leftskip=0.4em]{ 
U = U' \cup \{e{=}e'\} \\
E = E' \cup \{e \mapsto (i,m), e' \mapsto (i',m')\} \\
E_{in} = \{ e_s \;|\; (e_s \rightarrow e')\in D \} \\
E_{out} = \{ e_t \;|\; (e' \rightarrow e_t)\in D \} \\
D_{in} = \{ (e_s {\rightarrow} e') \;|\; e_s {\in} E_{in} \} \\
D_{out} = \{ (e' {\rightarrow} e_t) \;|\; e_t {\in} E_{out} \} \\
D = D' {\cup} D_{in} {\cup} D_{out} \\
I_s = \{k \;|\; E_{id}(S_k)\in\{E_{id}(e),E_{id}(e') \}\} \\
I_t = \{k \;|\; E_{id}(T_k)\in\{E_{id}(e),E_{id}(e') \}\} \\
}{ 
\langle \emptyset, E, (D, S, T), W, U, k \rangle \longrightarrow \langle \emptyset, E'{\cup}\{e{\mapsto}(i,\mathit{min}(m,m')),e'{\mapsto}(i,\mathit{min}(m,m'))\},
\allowbreak \\  (D'{\cup} \{e_s {\rightarrow} e \;|\; e_s {\in} E_{in}\}{\cup} \{e {\rightarrow} e_t \;|\; e_t {\in} E_{out}\}, \allowbreak \\ S{-}(I_s{-}\mathit{min}(I_s)), T{-}(I_t{-}\mathit{min}(I_t))), W, U', k \rangle
}
\vspace{-0.5em} \\

\inferrule*[Right={\scriptsize Union2},leftskip=0em,rightskip=2.5em]{ 
	E_{id}(S_0) = E_{id}(T_k) \\
}{ 
	\langle \emptyset, E, O, W, U, k \rangle \longrightarrow \langle E_{min}(S_0), E, O, W, U', k \rangle
}
\vspace{-0.5em} \\

\inferrule*[Right={\scriptsize Un.3},leftskip=1em]{ 
U {=} U' {\cup} \{e{\not=}e'\} \\
D {=} D' \cup \{ e_s \rightarrow e_t \;|\; \{E_{id}(e),E_{id}(e')\}=\{E_{id}(e_s),E_{id}(e_t)\} \} \hspace{2em}\\
1 \geq n \geq 0 \\
\{E_{id}(e),E_{id}(e')\} {=} \{E_{id}(S_0),E_{id}(T_k)\} \iff n{=}1 \\
}{ 
	\langle \emptyset, E, (D,S,T), W, U, k \rangle \longrightarrow \langle \emptyset, E, (D',S,T), W, U', k+n \rangle
}
\vspace{-0.25em} \\

\HRule \vspace{-0.5em}\\
	\fbox{Unifier} \hfill \vspace{-2.0em} \\

\inferrule*[Right={\scriptsize Eq.},leftskip=0em]{ 
	T = [\ldots, e_t, \ldots] \\
	D = D' \cup \{ e_s \xrightarrow{c} e_t \} \\
	c \leq \textit{min}(\{j \;|\; (x\xrightarrow{j}y) \in D'\}) \\
	e \in \mathit{class}(E, e_s)\\
	e' \in \mathit{class}(E, e_t)\\
	\semantics{e}\bisim{}\semantics{e'}
}{ 
	\langle \emptyset, E, (D, S, T), W, U, k \rangle \longrightarrow \langle \emptyset, E, (D', S, T), W\cup\{e \leftrightarrow e'\}, U\cup\{e{=}e'\}, k \rangle
}
\vspace{-0.5em} \\

\inferrule*[Right={\scriptsize Ineq.},leftskip=1em]{ 
	T = [\ldots, e_t, \ldots] \\
	D = D' \cup \{ e_s \xrightarrow{c} e_t \} \\
	c \leq \textit{min}(\{j \;|\; (x\xrightarrow{j}y) \in D'\}) \\
	e \in \mathit{class}(E, e_s)\\
	e' \in \mathit{class}(E, e_t)\\
	\semantics{e}\not\bisim{}\semantics{e'}
}{ 
	\langle \emptyset, E, (D, S, T), W, U, k \rangle \longrightarrow \langle \emptyset, E, (D', S, T), W, U\cup\{e{\not=}e'\}, k \rangle
}
\vspace{-0.5em} \\

\inferrule*[Right={\scriptsize Timeout},leftskip=3em]{ 
	T = [\ldots, e_t, \ldots] \\
	D = D' \cup \{ e_s \xrightarrow{c} e_t \} \\
	c \leq \textit{min}(\{j \;|\; (x\xrightarrow{j}y) \in D'\}) \\
	e \in \mathit{class}(E, e_s)\\
	e' \in \mathit{class}(E, e_t)\\
	\semantics{e}\,?\hspace{-3pt}\bisim{}\semantics{e'}
}{ 
	\langle \emptyset, E, (D, S, T), W, U, k \rangle \longrightarrow \langle \emptyset, E, (D'\cup\{e_s \xrightarrow{2\cdot c} e_t\}, S, T), W, U, k \rangle
}
\vspace{-0.25em} \\

\HRule \vspace{-0.25em}\\
	\fbox{Enumerator} \hfill \vspace{-2.0em} \\

\inferrule*[Right={\scriptsize Enum.},leftskip=2em]{ 
	T = [\ldots, e_n] \\
	e \in \mathit{min}(\{e' \;|\; \cost{E_{min}(S_0)} \geq \cost{e'} \geq \cost{e_n}\})\\
}{ 
	\langle \emptyset, E, (D, S, T), W, U, k \rangle \longrightarrow \langle \emptyset, E, (D\cup\{e_s \xrightarrow{1} e' \;|\; e_s \in S \}, S,T{+}[e]), W, U, k \rangle
}
\vspace{-0.5em} \\

\inferrule*[Right={\scriptsize Source},leftskip=0em]{ 
	S = [\ldots, e', \ldots] \\
	e \in \mathit{subexprs}(e')\\
}{ 
	\langle \emptyset, E, (D, S, T), W, U, k \rangle \longrightarrow \langle \emptyset, E, (D\cup\{e \xrightarrow{1} e_t \;|\; e_t \in T \}, S{+}[e], T), W, U, k \rangle
}

\end{mathpar}
}
\vspace{1em}
\caption{ReGiS as an abstract machine with $\longrightarrow$ denoting transitions.
Machine state is $\langle X, E, O, W, U, k\rangle$:
$X$ is a set containing a minimal regular expression upon termination; $E$ is the E-graph; $O = (D,S,T)$ is
the overlay graph with set of edges $D$, list of sources $S$, and list of targets $T$; $W$ is the set of rewrite rules;
$U$ is a set of (in)equalities to be processed; and $k$ is the index of the minimal unprocessed target in $T$. \hspace{5.0in}
$^\dagger$ {\sc Saturate} is used only in contexts where the rewrite rules have a completeness result.
}
\label{fig:regis_algo}
\end{figure*}

\subsection{ReGiS Algorithm}
\label{subsec:main-loop}

Figure \ref{fig:regis_algo} formalizes
ReGiS as an \emph{abstract machine} \cite{DBLP:journals/tcs/BerryB92}. 
A rule of the form $\frac{C}{S \longrightarrow S'}$ can be applied to step the machine state from $S$ to $S'$
if the condition $C$ is satisfied.
The algorithm terminates when no further steps can be taken.
A machine state is of the form $\langle X, E, O, W, U, k\rangle$,
where $X$ is a set used for storing the global
%
%
minimum (return value);
$E$ is the E-graph;
$O$ is a tuple $(D, S, T)$ representing the overlay graph, where $S$ and $T$ are the lists of source/target expressions respectively (overlay graph nodes), and $D$ is a set of weighted overlay graph edges;
$W$ is the set of rewrite rules;
$U$ is a set of expression (in)equalities to be incorporated into the E-graph; and
$k$ is the index of the minimal unprocessed target in $T$, i.e., lowest-cost target that has not yet been (in)equality-checked against
the source (in the Figure \ref{fig:overview} example, this would be the target at index 1).
Given optimization instance $(e,\costfn,W,\semantics{\cdot},\bisim{})$,
we use initial machine state $\langle \emptyset, E, (\emptyset, [e], \emptyset), W, \emptyset, 0 \rangle$, and run the machine until $X$ becomes non-empty, which causes the machine
to halt (the expression
contained in $X$ is the global minimum returned by the algorithm).
If the user prematurely terminates the machine,
we can output the current minimum $E_{min}(S_0)$, which may have lower cost than the source expression $S_0$,
but may not yet be the global minimum.

\subsection{Updater}
\label{subsec:updater}

The Updater's functionality is described
in Figure \ref{fig:regis_algo}
by the {\sc Rewrite}, {\sc Saturate}, and {\sc Union{}$\textit{X}$} rules.
Conceptually, the Updater functions as a wrapper for a persistent instance of Egg's E-graph data structure,
which is denoted $E$.
{\sc Rewrite} allows a single rewrite rule $w$ that matches an expression $e$ to be applied, and adds
the resulting equality $e{=}e'$ to the set $U$ to be incorporated into the E-graph via the {\sc Union$\textit{X}$} rules.
{\sc Saturate} is for cases where the rewrite rules have a completeness result.
This rule tests for a \textit{saturated} E-graph, i.e., in which none of the rewrite
rules change $E$---in this case, the algorithm can terminate (returning the global minimum $E_{min}(S_0)$), since
all possible rewrites have been explored.
The {\sc Union2} rule allows the algorithm to terminate when the minimal unprocessed target $T_k$
has been added to the source expression's E-class, since the enumeration order ensures that
$T_k$ will contain a globally-minimal expression.
{\sc Union1} incorporates an equality $e{=}e'$ into the E-graph $E$, by
(1) placing $e'$ into $e$'s E-class within $E$,
(2) updating the overlay graph edges $D$ by moving incoming/outgoing edges from $e'$ to $e$,
and (3) updating the overlay graph source/target lists by keeping only the \textit{lowest}
index belonging to the same E-class as $e$ or $e'$,
ensuring that the source ($S_0$) and
minimal unprocessed target ($T_k$) expressions are not absorbed into other sources/targets.
{\sc Union3} incorporates an inequality $e{\not=}e'$ obtained from the Unifier, by deleting any corresponding edges from the overlay graph.
If $e$ and $e'$ are contained in the source ($S_0$) and minimal unprocessed target ($T_k$) E-classes respectively (or vice versa),
{\sc Union3} additionally increments the index of the minimal unprocessed target $T_k$.

\subsection{Unifier}
\label{subsec:unifier}

Each spawned Unifier selects a minimal-weight source/target edge $e_s \xrightarrow{c} e_t$
from the overlay graph, and
performs a semantic equality check, as shown in
{\sc Equality}, {\sc Inequality}, and {\sc Timeout} in Figure \ref{fig:regis_algo}.
%
A single member from each class is selected ($e$ and $e'$),
and the equality check $\semantics{e}\bisim{}\semantics{e'}$ is performed.
If the check \textit{succeeds} ({\sc Equality}),
the Unifier records $e$ and $e'$ as needing to be unioned,
and a new rewrite rule $e \leftrightarrow e'$ is generated.
%
If the check \textit{fails} ({\sc Inequality}),
the Unifier records inequality $e{\not=}e'$, which the Updater will use
to delete the overlay graph edge, ensuring that this particular equality check is not attempted again.
If the equality-checking procedure \textit{times out} ({\sc Timeout}), the Unifier increases the expense estimate of the edge
joining $e_s$ and $e_t$, ensuring that other potentially-easier equality checks are tried before
returning to this pair.
Intuitively, timeout of the equality check means that we do not (yet) know whether the two expressions are semantically equivalent.

\subsection{Enumerator}
\label{subsec:enumerator}

The Enumerator's goal is to iterate through expressions in order of increasing cost, and
add them as targets (along with new overlay graph edges), as
shown in the {\sc Enumerate} rule.
One key issue with in-order enumeration is the sheer number of expressions involved, which
can be in the millions for even depth-4 binary trees.
This causes slowdown of the Enumerator itself, and makes it unlikely that higher-cost
expressions will be reached in a reasonable amount of time.
Rewriting helps address this problem, by maintaining the \textit{minimum expression} $E_{min}(S_0)$ in
the source expression's E-class.
As this expression's cost gets reduced by rewriting, $\cost{E_{min}(S_0)}$ becomes the new \textit{upper bound}
for cost within the Enumerator, reducing the search space, and speeding up enumeration.

Our Enumerator is powered by the Z3 SMT solver \cite{DBLP:conf/tacas/MouraB08}.
We require the cost function to be representable in SMT using a decidable theory such as QF\_UFLIA (quantifier-free uninterpreted function symbols and linear integer arithmetic), and
we use uninterpreted function symbols to encode
expressions as trees up to a specified maximum height, representing the cost metric via
assertions that maintain the cost of each tree node.
Section \ref{subsec:smt_cost}  contains encoding details for regular expressions.

In addition to adding  targets, the Enumerator can also add additional source expressions ({\sc Source} rule), to facilitate equality checking against \textit{subexpressions} of the source.

\subsection{Properties of ReGiS}
\label{subsec:properties}

\def\thmsoundness{
    If ReGiS returns an expression $e'$ for an optimization instance
$(e,\costfn,\allowbreak W,\semantics{\cdot},\bisim{})$,
    then
    $e'$ is no larger than the minimal expression in $hl(e)$, i.e.,  
    $\cost{e'} \leq \cost{e''}$ for any $e''\in hl(e)$.
}
\begin{theorem}[Soundness]\label{thm:soundness}
	\thmsoundness
\end{theorem}

\def\thmcompleteness{
    If $e_m$ is a minimal expression in $hl(e)$, i.e.,  
    $e_m \in hl(e)$ and
    $\cost{e_m} \leq \cost{e''}$ for any $e''\in hl(e)$,
    then ReGiS's result
    $e'$ will have $\cost{e'} = \cost{e_m}$.
}
\begin{theorem}[Completeness]\label{thm:completeness}
	\thmcompleteness
\end{theorem}

\noindent
The proofs of these theorems appear in Appendix \ref{subsec:proofs}.

\section{Regular Expression Optimization\pages{(5 pages)}}
\label{sec:simp}
\label{sec:redos}

In Section \ref{sec:regis}, we formalized the ReGiS framework, and in this section,
we highlight the flexibility and practicality of our approach by using it to solve
the important and insufficiently-addressed problem of optimizing superlinear regular expressions.

\subsection{Regular Expression Preliminaries}
\label{subsec:regex}
\label{subsec:prelim}
\label{subsec:regex_prelim}

{While $\S$\ref{subsec:expr_opt} discussed expressions generally, we will now focus specifically on \textit{regular expressions}.}
Regular expressions are a classic formalism providing a compositional syntactic approach for describing
\textit{regular languages}, and are useful for tokenizing input streams (e.g., in a lexer),
pattern matching within text, etc. In software development practice, the term ``regular expression''
is often overloaded to refer to a variety of pattern-matching capabilities and syntaxes,
so it is important to fix this definition for our work.
We say that an expression $R$ is a \textit{regular expression} if and only if
it matches the following grammar
(we will use the term \textit{regex} when referring to pattern-matching expressions beyond this core language).

\plstxset{or=\big\vert,or skip=2.5pt}
{
\begin{plstx}
(regular expression)         : R ::= 0 | 1 | c | R+R | R\cdot R | R^{\ast} \\
	*(character from alphabet)      : c [\in] \mathcal{A} \\
\set{or skip=1pt,or=\,}
\set{or=\vert}
\end{plstx}}

\noindent
An expression $R_1{\cdot}R_2$ is often written as $R_1 R_2$. Semantics can
be defined in terms of the language each regular expression recognizes.

\begin{displaymath}
\begin{array}{lll}
	\lang{0} & = & \emptyset \\
	\lang{1} & = & \{ \epsilon \} \\
	\lang{c} & = & \{ c \} \\
	\lang{R_1+R_2} & = & \lang{R_1} \cup \lang{R_2} \\
	\lang{R_1\cdot R_2} & = & \{ s_1 s_2 \;|\; s_1 \in \lang{R_1} \text{ and } s_2 \in \lang{R_2} \}\\
	\lang{R^{\ast}} & = & \bigcup_{k=0}^\infty \lang{R^k} \\
\end{array}
\end{displaymath}

\noindent
We use $R^k$ to mean $\overbrace{R{\cdot}R{\cdot}R{\cdots}}^{k}$ for $k > 0$, and $R^0 = 1$.
The above semantics tells us that
$0$ recognizes no strings,
$1$ recognizes the empty string $\epsilon$,
character $c$ recognizes the corresponding single-character string,
alternation $+$ recognizes the union of two languages,
concatenation $\cdot$ recognizes string concatenation,
and iteration (Kleene star) $\ast$ recognizes repeated concatenation.

\subsection{Regular Expression Semantic Equality}
\label{subsec:regex_equality}

To perform the Section \ref{subsec:unifier} (Unifier) semantic equality check $\bisim{}$ for regular expressions
$R_1$ and $R_2$, we must decide whether $\lang{R_1}  = \lang{R_2}$, i.e., whether they describe the same language.
Using Thompson's construction, we can efficiently convert $R$ to an NFA $N(R)$ which recognizes the language $L(R)$, and this result allows us to
instead focus on the equality check $L(N(R_1)) = L(N(R_2))$.
NFA equality is PSPACE-complete \cite{DBLP:conf/popl/MayrC13}, but we utilize a
bisimulation-based algorithm which has been shown to be effective in many cases \cite{DBLP:journals/jalc/AlmeidaMR10, DBLP:conf/setta/FuDJ017}.

\subsection{Regular Expression Syntactic Rewriting}
\label{subsec:regex_rules}
\label{subsec:regex_rewrite}

Regular expressions have a mathematical
formalization known as 
\textit{Kleene algebra}.
Due to soundness and completeness properties, equality of regular expressions can be fully
characterized by a set of Kleene algebra axioms.
These axioms leave us with two equally powerful
ways checking equality of regular
expressions $R_1, R_2$:
we can either check whether they describe the same language, i.e.,
$L(R_1)=L(R_2)$ (semantic equality) as discussed in Section \ref{subsec:regex_equality}, or we can 
check whether there is a \textit{proof} using the
Kleene algebra axioms showing that $R_1 = R_2$ (syntactic equality).

Following \citet{DBLP:conf/lics/Kozen91}, we list the Kleene algebra axioms as
follows, where $R_1 \leq R_2$ is shorthand for $R_1 + R_2 = R_2$.
{
$$
\begin{array}{r@{\hskip 1em}l@{\hskip 1em}l@{\hskip 2em}ll}
	A + (B + C) & = & (A+B)+C & \text{associativity of $+$} & (1) \\
	A+B & = & B+A & \text{commutativity of $+$} & (2) \\
	A+0 & = & A & \text{identity for $+$} & (3) \\
	A+A & = & A & \text{idempotence of $+$} & (4) \\
	A\cdot(B\cdot C)& = & (A\cdot B)\cdot C & \text{associativity of $\cdot$} & (5) \\
	1\cdot A & = & A & \text{left identity for $\cdot$} & (6) \\
	A\cdot 1 & = & A & \text{right identity for $\cdot$} & (7) \\
	A\cdot(B+C) & = & A{\cdot} B + A {\cdot} C & \text{left distributivity of $\cdot$} & (8) \\
	(A+B)\cdot C & = & A{\cdot} C + B {\cdot} C & \text{right distributivity of $\cdot$} & (9) \\
	0\cdot A & = & 0 & \text{left annihilator for $\cdot$} & (10) \\
	A\cdot 0 & = & 0 & \text{right annihilator for $\cdot$} & (11) \\
	1+A{\cdot}A^{\ast} & \leq & A^{\ast} & \text{left unrolling of $\ast$} & (12) \\
	1+A^{\ast}{\cdot}A & \leq & A^{\ast} & \text{right unrolling of $\ast$} & (13) \\
	B + A{\cdot}X \leq X & \Rightarrow & A^{\ast}{\cdot}B \leq X & \text{left unrolling ineq.} & (14) \\
	B + X{\cdot}A \leq X& \Rightarrow & B{\cdot}A^{\ast} \leq X & \text{right unrolling ineq.} & (15) \\
\end{array}
$$
}

We can derive the following useful rules from the above axioms. These rules are
extraneous due to the soundness and completeness of 1-15, but they offer several
ways of quickly eliminating Kleene stars, which
as seen in Section \ref{subsec:regex_metric}, are a primary contributor toward
superlinear behavior.
{
$$
\begin{array}{r@{\hskip 1em}l@{\hskip 1em}l@{\hskip 2em}ll}
	1+A{\cdot}A^{\ast} & = & A^{\ast} & \text{strong left unrolling of $*$} & (16) \\
	1+A^{\ast}{\cdot}A & = & A^{\ast} & \text{strong right unrolling of $*$} & (17) \\
	A^{\ast}\cdot A^{\ast}& = & A^{\ast} & \text{idempotence of $*$} & (18) \\
	A^{\ast\ast}& = & A^{\ast} & \text{saturation of $*$} & (19) \\
	1^{\ast} & = & 1 & \text{iterated identity} & (20) \\
	0^{\ast} & = & 1 & \text{iterated annihilator} & (21) \\
\end{array}
$$
}
We encode equations 1-11, 16-17 (the stronger forms of 12-13), and 18-21 as rewrite rules---for each equation
$X = Y$, we produce the bidirectional rewrite rule $X \leftrightarrow Y$,
{and instantiate our Updater (\ref{subsec:updater}) with these rules.}
Note that equations 14-15 are not equalities like the others, meaning they are not readily
usable as rewrite rules. We omit these equations, noting that the ReGiS approach fully
supports incomplete sets of rewrite rules.

\subsection{Backtracking Regular Expression Matching Algorithms}
\label{subsec:backtracking}

Catastrophic backtracking behavior arises due to the \textit{nondeterminism} in the NFAs corresponding to regular expressions.
Figure \ref{fig:superlinear} showed a quantitative example of the PCRE engine's exponential behavior
on the regular expression
$\chr{a}^{\ast\ast}$ and input strings
$\chr{a}\chr{a}\ldots\chr{a}\chr{b}$.
In Figure \ref{fig:nfa}, we visualize how this occurs, by examining the paths taken through NFA $N(\chr{a}^{\ast\ast})$
as the PCRE matching algorithm attempts to match,
and observing how the number of iterations increases exponentially.
To match against the string $\chr{b}$, all transitions to the \textit{left} of computation tree node $t_1$ in Figure \ref{fig:nfa}(b)
must be explored, before concluding that $\chr{b}$ is not accepted
($7$ transitions).
Note that backtracking algorithms use lightweight \textit{memoization} to handle cycles in the NFA (skipped
transitions due to this behavior are indicated with dashed lines) \cite{RussCox}.
To match
\begin{figure}
	\centering
\begin{tikzpicture}[scale=0.55]
	\begin{pgfonlayer}{nodelayer}
		\node [state, style=new style 0] (1) at (-8, 5) {};
		\node [state, style=new style 0] (2) at (-5, 5) {};
		\node [state, style=new style 0] (6) at (1, 5) {};
		\node [state, style=new style 0] (7) at (-2, 5) {};
		\node [state, style=new style 0] (8) at (4, 5) {};
		\node [state, accepting, style=new style 0] (9) at (7, 5) {};
		\node [style=none] (10) at (-8.5, 6.5) {};
	\end{pgfonlayer}
	\begin{pgfonlayer}{edgelayer}
		\draw [style=new edge style 0] (1) to node [auto, style={edge_label}] {$\epsilon_1$} (2);
		\draw [style=new edge style 0] (2) to node [auto, style={edge_label}] {$\epsilon_5$} (7);
		\draw [style=new edge style 0] (7) to node [auto, style={edge_label}] {\chr{a}} (6);
		\draw [style=new edge style 0] (6) to node [auto, style={edge_label}] {$\epsilon_7$} (8);
		\draw [style=new edge style 0] (8) to node [auto, style={edge_label}] {$\epsilon_3$} (9);
		\draw [style=new edge style 0] (10.center) to (1);
		\draw [style=new edge style 0, bend left=285, looseness=1.25] (6) to node [above, style={edge_label}] {$\epsilon_8$} (7);
		\draw [style=new edge style 0, bend left=285] (8) to node [above, style={edge_label}] {$\epsilon_4$} (2);
		\draw [style=new edge style 0, bend right=75, looseness=0.50] (2) to node [auto, style={edge_label}] {$\epsilon_6$} (8);
		\draw [style=new edge style 0, bend right=60, looseness=0.75] (1) to node [auto, style={edge_label}] {$\epsilon_2$} (9);
	\end{pgfonlayer}
\end{tikzpicture}\hspace{0pt}\begin{tikzpicture}[scale=0.75]
	\begin{pgfonlayer}{nodelayer}
		\node [style={tree_node}] (0) at (-9, 8) {$t_0$};
		\node [style={tree_node}] (1) at (-7, 8) {};
		\node [style={tree_node}] (2) at (-5, 4) {};
		\node [style={tree_node}] (3) at (-8, 6) {};
		\node [style={tree_node}] (4) at (-5, 8) {};
		\node [style={tree_node}] (5) at (-6, 6) {};
		\node [style={tree_node}] (6) at (-4, 6) {};
		\node [style={tree_node}] (7) at (-3, 8) {$t_1$};
		\node [style={tree_node}] (8) at (-1, 8) {};
		\node [style={tree_node}] (9) at (-2, 6) {};
		\node [style={tree_node}] (10) at (1, 8) {};
		\node [style={tree_node}] (11) at (0, 6) {};
		\node [style={tree_node}] (12) at (-1, 4) {};
		\node [style={tree_node}] (13) at (2, 6) {};
		\node [style={tree_node}] (14) at (1, 4) {};
		\node [style={tree_node}] (15) at (4, 6) {$t_2$};
		\node [style={tree_node}] (16) at (6, 6) {};
		\node [style={tree_node}] (17) at (5, 4) {};
		\node [style={tree_node}] (18) at (6, 2) {};
		\node [style={tree_node}] (19) at (7, 4) {};
		\node [style={tree_node}] (20) at (8, 6) {};
		\node [style={tree_node}] (22) at (9, 4) {};
		\node [style={tree_node}] (23) at (10, 2) {};
		\node [style={tree_node}] (24) at (8, 2) {};
	\end{pgfonlayer}
	\begin{pgfonlayer}{edgelayer}
		\draw [style=new edge style 0] (0) to node [auto, style={edge_label}] {$\epsilon_1$} (1);
		\draw [style=new edge style 0] (0) to node [auto, style={edge_label}] {$\epsilon_2$} (3);
		\draw [style=new edge style 0] (1) to node [auto, style={edge_label}] {$\epsilon_5$} (4);
		\draw [style=new edge style 0] (1) to node [auto, style={edge_label}] {$\epsilon_6$} (5);
		\draw [style=accepted] (4) to node [auto, style={edge_label}] {$\chr{a}$} (7);
		\draw [style=new edge style 0] (5) to node [auto, style={edge_label}] {$\epsilon_3$} (6);
		\draw [style=new edge style 0] (7) to node [auto, style={edge_label}] {$\epsilon_7$} (8);
		\draw [style=new edge style 0] (8) to node [auto, style={edge_label}] {$\epsilon_3$} (10);
		\draw [style=new edge style 0] (8) to node [auto, style={edge_label}] {$\epsilon_4$} (11);
		\draw [style=new edge style 0] (7) to node [auto, style={edge_label}] {$\epsilon_8$} (9);
		\draw [style={accepted_dash}] (9) to node [auto, style={edge_label}] {$\chr{a}$} (12);
		\draw [style=new edge style 1] (5) to node [auto, style={edge_label}] {\color{black}{$\epsilon_4$}} (2);
		\draw [style=new edge style 1] (11) to node [auto, style={edge_label}] {\color{black}{$\epsilon_6$}} (14);
		\draw [style=new edge style 0] (11) to node [auto, style={edge_label}] {$\epsilon_5$} (13);
		\draw [style=accepted] (13) to node [auto, style={edge_label}] {$\chr{a}$} (15);
		\draw [style=new edge style 0] (15) to node [auto, style={edge_label}] {$\epsilon_7$} (16);
		\draw [style=new edge style 0] (15) to node [auto, style={edge_label}] {$\epsilon_8$} (17);
		\draw [style={accepted_dash}] (17) to node [auto, style={edge_label}] {$\chr{a}$} (18);
		\draw [style=new edge style 0] (16) to node [auto, style={edge_label}] {$\epsilon_3$} (20);
		\draw [style=new edge style 0] (16) to node [auto, style={edge_label}] {$\epsilon_4$} (19);
		\draw [style=new edge style 0] (19) to node [auto, style={edge_label}] {$\epsilon_5$} (22);
		\draw [style=accepted] (22) to node [auto, style={edge_label}] {$\chr{a}$} (23);
		\draw [style=new edge style 1] (19) to node [auto, style={edge_label}] {\color{black}{$\epsilon_6$}} (24);
	\end{pgfonlayer}
\end{tikzpicture}
	\caption{(a) NFA for regular expression $\chr{a}^{\ast\ast}$, and (b) its (partial) computation tree.}
	\label{fig:nfa}
\end{figure}
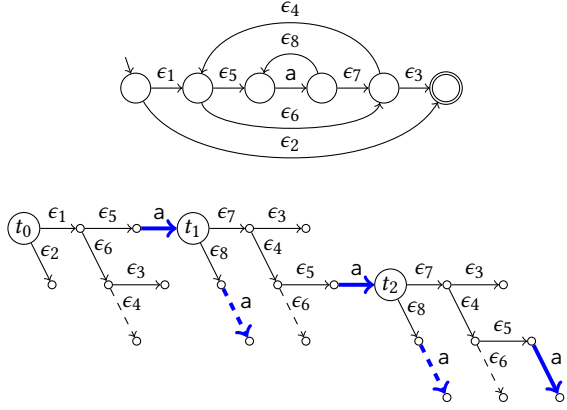
%
against the string $\chr{a}\chr{b}$, all transitions to the left of node $t_2$ must be explored
($15$ transitions), and to match against $\chr{a}\chr{a}\chr{b}$, all transitions in Figure \ref{fig:nfa}(b)
must be explored ($23$ transitions).
Here, superlinear behavior can be triggered
by a large sequence of
repeated $\chr{a}$ characters followed by a non-\chr{a} character.
In practice, one can automatically derive such an \textit{attack string} to
exploit a given superlinear regular expression \cite{Shen}.

\subsection{Cost Metric for Superlinear Regular Expressions}
\label{subsec:regex_metric}

An upper
bound on the maximum backtracking for a given regular expression
can be characterized by \textit{tree width (leaf size)} \cite{palioudakis2015quantifying,DBLP:conf/dcfs/CampeanuS15},
which describes the number of leaves in the tree
consisting of all possible paths through
the regular expression's NFA, but this is
hard to compute (PSPACE-complete) \cite{DBLP:conf/dcfs/CampeanuS15}.
In Figure \ref{fig:nfa}, tree width would be $10$ with respect to the
input string $\chr{a}\chr{a}\chr{b}$, since this (depth-3) computation tree has $10$ leaves.
Alternative metrics such as maximal backtracking
run \cite{Weideman} quantify potential backtracking in different
ways, but are also computationally expensive.
We introduce a useful cost
metric which we call \textit{backtracking factor}
that is quick to compute directly on a regular
expression, yet still captures the key
syntactic features causing superlinearity
in backtracking search.
This backtracking factor allows ordering of regular expressions
according to ``degree of superlinearity'',
e.g., $\cost{\chr{a}} < \cost{\chr{a}^{\ast}} < \cost{\chr{a}^{\ast\ast}}$.

The following shows our cost metric, where
$K_1 = |\mathcal{A}|\times(2^h - 1)$
and
$K_2 = {K_1}^h \times (K_1+2)$ are integer scaling
factors for $\cdot$ and $\ast$ respectively, $h$ is the maximum expression
height being used in the Enumerator, and $|\mathcal{A}|$ is the size
of the regular expression's alphabet.

\begin{displaymath}
\begin{array}{lll}
	\cost{0} & = & 1 \\
	\cost{1} & = & 1 \\
	\cost{c} & = & 1 \\
	\cost{R_1+R_2} & = & \cost{R_1} + \cost{R_2} \\
    \cost{R_1\cdot R_2} & = & K_1 \times (\cost{R_1} + \cost{R_2}) \\
    \cost{R^{\ast}} & = & K_2 \times \cost{R} \\
\end{array}
\end{displaymath}

This metric has the effect of ensuring that, for regular expressions up to a maximum height
$h$, the $+$ operator increases cost \textit{additively},
the $\cdot$ operator increases cost \textit{multiplicatively}, and $\ast$ increases cost \textit{exponentially}.
{We instantiate our Enumerator (Section \ref{subsec:enumerator}) with this specific cost metric, and}
in Section \ref{sec:eval}, we experimentally validate the quality of this metric for
characterizing superlinear behavior.

The following property of \costfn{} says that for the lowest-cost expression $R_2$ seen so far,
the regular expression with \textit{globally} minimal cost will have height no greater than
$\expheight{R_2}$.
This ensures that we can soundly reduce the Enumerator's \textit{height} bound
based on the height of the current lowest-cost expression in the source expression's E-class,
which improves enumeration performance.

\def\thmheightcostpre{
Consider regular expressions $R_1,R_2$.
If $L(R_1) \allowbreak = L(R_2)$, \allowbreak $\expheight{R_1}\allowbreak >\expheight{R_2}$,}
\def\thmheightcostpost{and $\cost{R_1} \allowbreak\leq \cost{R_2}$,
then $\exists\, R'$ such that $L(R')=L(R_2)$, $\expheight{R'}\leq \expheight{R_2}$, and
$\cost{R'} \leq \cost{R_1}$.
}
\def\lemheightcost{
	\thmheightcostpre{} $|\{R \;|\; L(R_1)=L(R)=L(R_2) \text{ and } \expheight{R_1} > \expheight{R} > \expheight{R_2}\}| = 0$,
\thmheightcostpost
}
\def\thmheightcost{
	\thmheightcostpre{} \thmheightcostpost
}

\begin{theorem}[Height vs. Cost]\label{thm:heightcost}
\thmheightcost
\end{theorem}

\begin{figure}[t]
\centering
	\begin{minipage}{0.4\linewidth}
\begin{displaymath}
\begin{array}{lll}
	 & \vdots &  \\
	\cost{c} & = & 1 \\
	\cost{R_1+R_2} & = & \cost{R_1} + \cost{R_2} \\
	 & \vdots &  \\
\end{array}
\end{displaymath}
	\end{minipage}
	\caption{Example cost function.}
	\label{fig:smt_cost_function}
    \hspace{-1.5in}
	\begin{minipage}{0.6\linewidth}
	{
	\begin{lstlisting}[mathescape,numbers=none,lineskip=0.0em,basicstyle={\scriptsize\ttfamily}]
(assert (= (ncost $N$)
	   (ite (= (ntype $N$) $\mathit{CHAR}$) 1
	   (ite (= (ntype $N$) $\mathit{PLUS}$) (+ (ncost $(\mathit{left}\; N)$) (ncost $(\mathit{right}\; N)$))
	   ...))))
\end{lstlisting}
	}
	\end{minipage}
	\caption{SMT encoding.}
	\label{fig:smt_encoding}
\end{figure}

\subsection{SMT Implementation of Cost Metric}
\label{subsec:smt_cost}

We allow the cost metric $\costfn$ to be specified as a recursive integer-valued function using addition as well as multiplication by integer constants
(the Section \ref{subsec:regex_metric} cost metric is of this form). This allows us to implement $\costfn$ in Egg as a recursive Rust function, used
for extracting the minimal expression from an E-class.
Additionally, this allows us to implement $\costfn$ in SMT using the QF\_UFLIA theory, to enable enumerating regular expressions by increasing cost, as needed by the Enumerator (Section \ref{subsec:enumerator}).

Our SMT encoding uses uninterpreted function symbols $\mathit{ntype}: \mathbb{N} \rightarrow \mathbb{N}$ and $\mathit{ncost}: \mathbb{N}\rightarrow\mathbb{N}$,
representing the type and cost of each node in the expression's tree.
Figures \ref{fig:smt_cost_function}-\ref{fig:smt_encoding} demonstrate
how our regular expression cost function is encoded using these function symbols.
The uppercase symbols in the SMT encoding signify integer constants---e.g., $\mathit{CHAR}$ identifies a character expression node type, and
$\mathit{PLUS}$ identifies an alternation expression node type. We generate one such
assertion for each node index $N$, up to the bounds given by the source expression's height.
Integer constants $(\mathit{left}\; N)$ and $(\mathit{right}\; N)$ are the indices of the nodes corresponding to node $N$'s left and right subexpressions respectively.

In the model obtained from the solver, the $\mathit{ntype}$ function symbol encodes the expression itself,
and $(\mathit{ncost}\; 0)$ contains the expression's cost (index 0 corresponds to the expression's root node).

\section{Prototype and Evaluation \pages{(4 pages)}}
\label{sec:eval}

We built a prototype of ReGiS ($\S$\ref{sec:regis}), and leveraged it to build a regular expression optimization
system ($\S$\ref{sec:redos}) 
using $\sim$7500 lines of Rust code and several hundred lines of Python/shell script.

The platform used for all experiments was a Dell OptiPlex 7080 workstation running Ubuntu 18.04.2,
with a 10-core (20-thread) Intel i9-10900K CPU (3.70GHz), 128 GB DDR4 RAM,
and a 2TB PCIe NVME Class 40 SSD.
We examine three key research questions Q1-3,
to understand the performance and usability of our approach.
{Note that we have proven the correctness of our algorithm (Section \ref{subsec:properties}), but to confirm that the implementation is bug-free, we used the equality checking procedure (Section \ref{subsec:regex_equality}) to successfully check the correctness of each result produced by our tool in the experiments.}

\subsection*{Q1: How does ReGiS compare against SyGuS and rewriting?}

We demonstrate the benefits of \textit{combining} enumerative synthesis
and rewriting, by comparing ReGiS performance against each of these approaches operating on
their own.
We used the Enumerator to simply iterate through candidate expressions
in increasing order of cost, performing an NFA equality check against the source expression
for each candidate.
The input source expressions consisted of all possible regular expressions with a single-character
alphabet, up to height 3, for a total of 2777 inputs.
Figure \ref{fig:exp01} shows these performances results.

Figure \ref{fig:exp01}(a) contains many overlapping points, so
we binned the data (bin size $0.1$)
and used size/color of the points to indicate relative numbers of tests appearing at those locations.
The diagonal indicates 1x speedup, so points appearing \textit{above} this line indicate better
performance for ReGiS versus basic enumeration.
Figure \ref{fig:exp01}(b) visualizes the data differently, showing the spectrum of \textit{speedups}
offered by ReGiS. Each bar represents the number of tests
in which ReGiS had the speedup shown on the $x$ axis, e.g.,
$222$ tests had a speedup of $30$x.

The bar(s) to the left of $x=0$ contain $821$ examples. Of these, $644$ are due to timeout of both \textit{both} ReGiS and
basic enumeration (time limit $3$s), and the remaining $177$ were instances where ReGiS was slower.
Of these, the average slowdown was $1.06$x, and the maximum slowdown was $2$x.
Only $10$ cases were worse than $1.15$x slowdown, and in each of these cases, the total runtime of ReGiS
was less than $110$ms.
We are confident that if the timeout were increased slightly, many of the $644$ examples would show speedup for ReGiS.

We also added an Egg-only mode to enable rewriting without enumeration. Here,
we simply added the input source expression to the E-graph, and applied rewrites
until Egg indicated saturation had been reached.
All of the input regular expressions timed out at $3$s using this rewriting-only mode.

\begin{figure}[tb]
\centering
{\includegraphics[trim = 0in 0in 0in 0in, clip,width=0.75\linewidth]{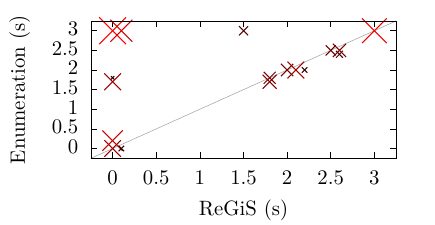}}\\%
{\includegraphics[trim = 0in 0in 0in 0in, clip,width=0.75\linewidth]{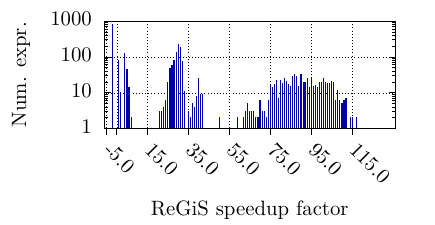}}
\caption{Performance of ReGiS against basic bottom-up enumeration.}
\label{fig:exp01}
\end{figure}

\subsection*{Q2: How much does the \emph{interplay} between enumeration and rewriting help?}

One key question is whether enumeration and rewriting could
be decoupled while still obtaining the same results. For example, we could first let
Egg perform some rewriting, and after we notice that no further reduction in cost seems to be
occurring, terminate Egg, and begin enumerating based on the best-cost expression so far.
If the performance of this approach were comparable to ReGiS, that would mean our work's
unique interaction between rewriting and synthesis may be less important than expected.
We set up a simple experiment similar to the alternation example described in Section \ref{subsec:limitations}, 
consisting of depth-4 regular expressions each having $5$ distinct
characters, and using only alternation. In this case, (1) the expressions were
not reducible, i.e., we must always enumerate up to the maximum depth to confirm
the global minimum has been found, and (2) there were many \textit{equivalent} expressions
involved in enumeration, due to commutativity of alternation.

\begin{table}[t]
	\centering
	\caption{Reduction in needed semantic equality checks.}
	\label{table:q2}
\begin{tabular}{|l||lc|lc|} \hline
	& \textbf{ReGiS} & &\textbf{Enumeration}&	\\\hline\hline
		& Checks	& Runtime (s)	& Checks	& Runtime (s) \\\hline
	 Min.	& 3067	& 11.80	& 3416	& 11.09 \\
	 Mean	& 3134.44	& 12.11	& 3416	& 11.40 \\
	 Max.	& 3235	& 12.33	& 3416	& 11.86\\\hline
\end{tabular} \\\vspace{1em}
    \vspace{-1em}
\end{table}

Table \ref{table:q2} shows the results on a benchmark set
of $100$ depth-4 regular expressions. There was reduction 
in the number of \textit{equality checks} needed by ReGiS---%
this is because the E-graph is continuously unioning candidate expressions added
from the
Enumerator, meaning that by the time a Unifier would be spawned for a given pair
of E-classes, they may have already been handled by rewriting.
Total ReGiS runtime was slightly worse on these particular examples,
due to overhead involved in maintaining the various data structures
in our prototype implementation.
The fact that our approach allows
equality checks to be skipped is vital in other domains
where the checker/verifier may be much slower than NFA bisimulation.

\subsection*{Q3: How does ReGiS compare with existing regex optimizers?}

Although existing regex optimizers may terminate quickly, they do not guarantee minimality of the result regular expressions. 
We show that our regular expression
optimization tool produces high-quality results compared to existing tools,
thereby validating design choices such as our regular expression cost metric.
We identified several existing open-source regex optimization tools,
\textit{Regexp-Optimizer} \cite{RegexOpt} (we will refer to this as \textit{Opt03}),
\textit{RegexOpt}, \cite{RegexpOptimizer} (we will refer to this as \textit{RegOpt}), and
\textit{Regular Expression Gym}, a part of the \textit{Noam} project \cite{Noam},
and installed them locally on our workstation.
We needed a set of ``ground truth'' regular expressions for which we know the minimum-cost equivalent
expressions.
Thus, we selected the regular expressions from Q1
for which our tool reported a minimum, giving us 2176 inputs, and from these 
we selected only those containing at least one Kleene star, giving us 1574.
We then transformed each of these into an equivalent expression
known to be superlinear---we used the RXXR regular expression static analysis tool \cite{Kirrage,DBLP:journals/corr/RathnayakeT14}
to check whether an expression was vulnerable, and if not, we randomly applied semantics-preserving
transformations known to increase complexity, e.g., $\chr{a}^{\ast} \rightarrow \chr{a}^{\ast}\chr{a}^{\ast}$.
This was repeated until RXXR confirmed vulnerability.
Finally, RXXR provided us with an attack string,
allowing us to target the vulnerability.

Using these superlinear expressions, we ran ReGiS and the
other optimizers.
The output expression from each was given to the PCRE regex engine,
and matched against the respective attack string.
{We found that the number of PCRE steps is generally proportional to PCRE runtime. Ultimately, the number of steps provides a more ``implementation-independent'' measure of regular expression complexity than runtime, since it roughly corresponds to number of steps in the NFA traversal (Section \ref{subsec:backtracking}).}
Figure \ref{fig:exp05} shows
the number of PCRE steps on the respective regular expressions.
Most points appear above the diagonal, meaning ReGiS produced
lower-cost regular expressions in terms of PCRE matching complexity.

\begin{figure}[tb]
\centering
{\includegraphics[trim = 0in 0in 0in 0in, clip,width=0.75\linewidth]{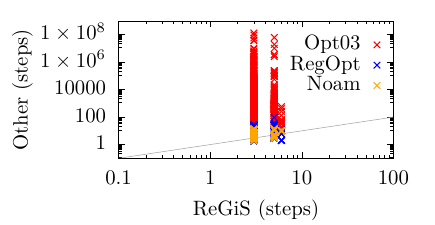}}\\%
{\includegraphics[trim = 0in 0in 0in 0in, clip,width=0.75\linewidth]{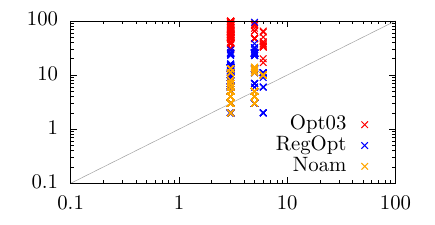}}
	\caption{Quality of ReGiS results versus open-source regex optimizers: (a) Number of PCRE steps to match against attack string; (b) Detail near diagonal.}
\label{fig:exp05}
\end{figure}

\section{Discussion and Future Work \pages{(1 page)}}
\label{sec:discussion}

While our benchmarks are small, they are exhaustive in the sense that  \textit{all expressions} having the given structure/bounds are included.
Tools like RXXR \cite{Kirrage,DBLP:journals/corr/RathnayakeT14} accept an input regex and identify a \textit{vulnerable subexpression} leading to superlinear behavior,
and this subexpression is often smaller than the input, meaning that handling small expressions has real-world value.
For example, in a 2793-regex dataset that the RXXR authors mined from the internet, RXXR identifies 122 regexes as vulnerable---in 99 of these, the vulnerable subexpression
has length 50 or less, even though the vulnerable input regexes have lengths up to 1067 (vulnerable subexpressions were up to 53x smaller than the inputs, averaging 3.9x smaller).

Even focusing on small exhaustive regular expression benchmarks, serious problems can arise
(Figure \ref{fig:superlinear} shows superlinear behavior of a single-character regular expression),
and we have shown that existing optimizers fail to offer workable solutions.

While we believe ReGiS to be an important step toward the high-level goal
of scaling up synthesis, especially in regards to regular expression optimization,
there are engineering and research
challenges we plan to address in future work.

\begin{compactitem}
\item Our regular expression optimization technique handles ``pure'' regular expressions, while many real-world regexes
go beyond this core language.
It is not fundamentally difficult to extend support, but this will require giving Egg more flexible rewrite-rule functionality, such as operations on character classes.
	Note that the primary ingredient of superlinear behavior is alternation under Kleene star, leading to exponential backtracking (Section \ref{subsec:backtracking})---no
	functionality outside pure regular expressions is needed to trigger this.

		\vspace{6pt}

\item Solver frameworks powered by DPLL rely on \textit{heuristics} to improve
search performance. There are similar opportunities for carefully-designed heuristics here.
For example, we assign Unifiers based on lowest overlay edge cost, but these could also be chosen
based on ``similarity'' of contained expressions, increasing likelihood of fast equality checks.
		\vspace{6pt}

\item Additional optimizations are possible. \textit{Verification} is often a bottleneck in synthesis---%
for us, this is an NFA equality check, which is usually fast,
but we could, e.g., parallelize several equality checks within each pair of E-classes, utilizing
the fastest result.
\textit{Incremental} node/expression cost maintenance in the E-graph would also improve performance.
We found that the Hopcroft-Karp (HK) NFA bisimulation algorithm which exploits equivalence classes does
not seem to offer improvement over a na\"ive on-the-fly NFA-to-DFA bisimulation check on the
DAG-like
Thompson NFAs.
We plan to investigate other equality-checking techniques \cite{DBLP:conf/popl/BonchiP13}.
		\vspace{6pt}

\end{compactitem}

\section{Related Work \pages{(2 pages)}}
\label{sec:rel_work}


\subsection*{Program Synthesis and Superoptimization}

\citet{DBLP:conf/cav/JeonQSF15} tackle the synthesis scalability problem by handling multiple
enumeration steps in parallel.
\citet{DBLP:conf/tacas/AlurRU17} use a divide-and-conquer approach, partitioning the set of inputs, solving a smaller
synthesis problem within each partition, and then combining the results together.
Superoptimization \cite{DBLP:conf/asplos/Schkufza0A13,DBLP:conf/asplos/PhothilimthanaT16} is an approach for optimizing sequences of instructions.
In contrast, we perform rewriting (syntactic) and enumeration (semantic) steps in parallel---the interplay
between these is key to our approach.

\subsection*{Combining Rewriting with Synthesis}

\citet{DBLP:conf/pldi/HuangQSW20} describe an approach which combines parallelism with a divide-and conquer methodology
to perform synthesis with enumeration and deduction (conceptually similar to rewriting). ReGiS offers additional parallelization
opportunities, by allowing the enumeration and rewriting to happen in parallel.

Using a rewriting-based approach for expression
optimization requires a technique for overcoming \textit{local minima}
in the rewriting. We achieve this by combining equality saturation-based rewriting
with syntax-guided synthesis.
%
\citet{DBLP:conf/pldi/NandiWAWDGT20} leverage
equality saturation \cite{DBLP:journals/corr/abs-1012-1802, DBLP:journals/pacmpl/WillseyNWFTP21}
while performing search for CAD model decompilation, but
they do not offer global minimality guarantees, or cost functions beyond expression size.
They overcome local minima by speculatively adding non-semantics-preserving rewrites, and ``undo'' these later, after the final expression has been extracted.

Cosy \cite{DBLP:conf/icse/LoncaricET18, DBLP:conf/pldi/LoncaricTE16} enumeratively synthesizes
data structures, using a \textit{lossy} ``deduplication'' mechanism
to maintain equivalence classes. Our approach compactly maintains \textit{all} equivalence
class members, not just representatives.
\citet{DBLP:conf/vmcai/SmithA19} combine
synthesis with rewriting, but require a \textit{term-rewriting
system} (TRS). Constructing a TRS is undecidable, so human input (using a proof
assistant) is typically needed, while our approach is fully-automated.

There are machine learning-based approaches that combine synthesis-like search with rewriting \cite{DBLP:conf/fmcad/0002S16, NeuRewriter}.
The key distinction between these approaches and ours
is the need for \textit{training data}, whereas our approach is designed to operate without this.
These data-driven approaches typically cannot guarantee \textit{global} optimality, and also cannot
propose \textit{new} rules---said another way, they offer a purely syntactic approach to optimization.


A related topic is \textit{synthesizing rewrite rules},
which has been investigated in the context of security hardware/software \cite{DBLP:conf/pldi/LeeLOY20}
and SMT \cite{DBLP:conf/sat/NotzliRBNPBT19}.
Another related direction is \textit{theory exploration}, which uses E-graphs to enumerate lemmas for 
theorem proving \cite{DBLP:journals/corr/abs-2009-04826}.

\balance

\subsection*{Regular Expression Denial of Service (ReDoS) Attacks}

Algorithmic complexity attacks are well-known, with early research in this area focusing on
network intrusion detection and other systems-related functions \cite{Crosby, Smith}.
Catastrophic backtracking and attacks against the complexity of regular expressions have also begun to appear in the literature \cite{Davis1},
but while some useful rule-of-thumb guides have been available for some time \cite{Catastrophic,AttacksAndDefenses,CheckMarx},
\textit{general awareness}
of these vulnerabilities may not be widespread.

\subsection*{Vulnerable Regular Expression Detection}

Existing work \textit{detects} regular expressions vulnerable to ReDoS.
\textit{Static analysis} can find the complexity of a regular expression, e.g.,
\citet{DBLP:journals/corr/BerglundDM14} formalize regular expression matching in Java, and statically
determine whether a given Java regular expression has exponential runtime.
\citet{Weideman} build on that work, providing more precise characterization of worst-case runtime.

A related problem is finding an \textit{attack string} that causes poor performance on
a given regular expression.
ReScue \cite{Shen} does this via genetic search and properties of the \textit{pumping lemma}.
RXXR \cite{Kirrage,DBLP:journals/corr/RathnayakeT14} finds an attack string, and a vulnerable \textit{subexpression}
causing superlinear behavior on that string.
Rexploiter \cite{Wustholz} constructs an \textit{attack automaton}, characterizing
the entire language of attack strings.
These approaches are complimentary to
ours, which seeks to \textit{remove vulnerabilities} from known-vulnerable regular expressions.

\subsection*{ReDoS Attack Prevention}

Some authors have suggested updating existing backtracking algorithms with a \textit{state cache} \cite{Davis2}, which
\textit{fully} memoizes traversal of the NFA, achieving polynomial runtime at the expense of significant memory usage.
Since developers may seek more power than core regular expressions provide, other work
seeks to extend the expressibility of Thompson-like approaches to support additional features
like backreferences \cite{DBLP:conf/infocom/NamjoshiN10}.

\textit{Synthesis from examples} \cite{DBLP:conf/pldi/Chen0YDD20,Pan} is related to our work.
This requires 
input-output examples, and an expression is automatically constructed to fit these examples.
The work does not directly address ReDoS, but could potentially be leveraged for that purpose.

\section{Conclusion}
\label{sec:conclusion}

We present rewrite-guided synthesis (ReGiS), a new approach for
expression optimization that interfaces
syntax-guided synthesis (SyGuS) with equality saturation-based rewriting.
We leverage ReGiS to address
the problem of optimizing superlinear regular expressions,
demonstrating the power and flexibility of our framework.

\section*{Acknowledgements}                            
We are grateful to the anonymous reviewers for their detailed comments.
This work is supported by the
\grantsponsor{GS1}{National Science
Foundation}{http://dx.doi.org/10.13039/100000001} under Grant
No.~\grantnum{GS1}{2018910} and Grant
No.~\grantnum{GS1}{2124010}.


\newpage

\printbibliography

\clearpage
\appendix

\section{Proofs of Theorems}
\label{subsec:proofs}

\begin{customthm}{thm:soundness}{Soundness}
\thmsoundness
\end{customthm}
\begin{proof}
	If the algorithm returns an expression $e'$, then the
	Figure \ref{fig:regis_algo} machine's final step must have been either {\sc Saturate} or
	{\sc Union2} (these are the only rules that allow termination).

	In the case of {\sc Saturate} (with a \textit{complete} rewrite rule set $W$),
	by definition of E-graph saturation, all
	possible equalities (modulo the rewrite rules $W$) have been incorporated
	into the E-graph $E$, meaning the source E-class $\mathit{class}(E, e)$
	contains all grammatically valid expressions that are
	syntactically equivalent to $e$.
	Let $e_m$ be a minimal expression in $hl(e)$. By completeness of $W$,
	since $\semantics{e_m} \bisim{} \semantics{e}$, we know that
	$e_m \in \mathit{class}(E, e)$.
	Thus, since the return value is $E_{min}(S_0) = E_{min}(e) = \mathit{min}(\mathit{class}(E, e))$, we know that $\cost{E_{min}(S_0)} \leq \cost{e_m}$.

	In the case of {\sc Union2}, we know that $S_0$ and the minimal unprocessed target $T_k$
	are equivalent, since they are in the same E-class.
	No target $T_j$ (where $j < k$) is equal to the source, since the index $k$ can
	only be incremented by {\sc Union3} when processing an inequality.
	Because {\sc Enumerate} adds all elements of $hl(e)$ to $T$ in increasing order of cost,
	we have	that $T_k$ is the lowest-cost element of $hl(e)$ that is equivalent to the source.
	Thus, $E_{min}(S_0)$ is no greater than the cost of the minimal element
	of $hl(e)$.
\end{proof}

\begin{customthm}{thm:completeness}{Completeness}
\thmcompleteness
\end{customthm}
\begin{proof}
	{\sc Enumerate} adds every expression with height no greater than $\mathit{height}(e)$
	to $T$ in increasing order of cost.
	Thus, the minimal element $e_m \in hl(e)$ will be added to $T$ at some index j,
	and $T_i \not\bisim{} e$ for all $i < j$.
	Because new overlay graph edges initially have unit weight, once $e_m$ appears as
	$T_j$, {\sc Equality} is able to register the equality $e_m = e$, and
	{\sc Union1} will be able to place $e_m$ and $e$ into the same E-class.
	For each target $T_i$ where $i < j$, {\sc Inequality} and {\sc Union3} will
	ensure that the index $k$ of the minimal unprocessed target $T_k$ is incremented,
	resulting in $k = j$.
	At this point, the {\sc Union2} rule will allow the algorithm to terminate
	with $E_{min}(S_0)$, which has cost no greater than $e_m$.
\end{proof}

\begin{lemma}\label{lem:heightcost}
\lemheightcost
\end{lemma}
\begin{proof}
	Figure \ref{fig:thm3} shows this graphically.
	There are no regular expressions in the vertical space between $R_1$ and $R_2$.
	The path connecting $R_1, R_2$ is formed by a sequence of rewrite rules.
	No segment (application of a single rewrite rule) can angle down and to the right
	or up and to the left,
	since no individual rewrite rule increases cost while reducing height (or vice versa).

	\begin{figure}[tb]
		\centering
{\includegraphics[trim = 0in 0in 0in 0in, clip,width=0.5\linewidth]{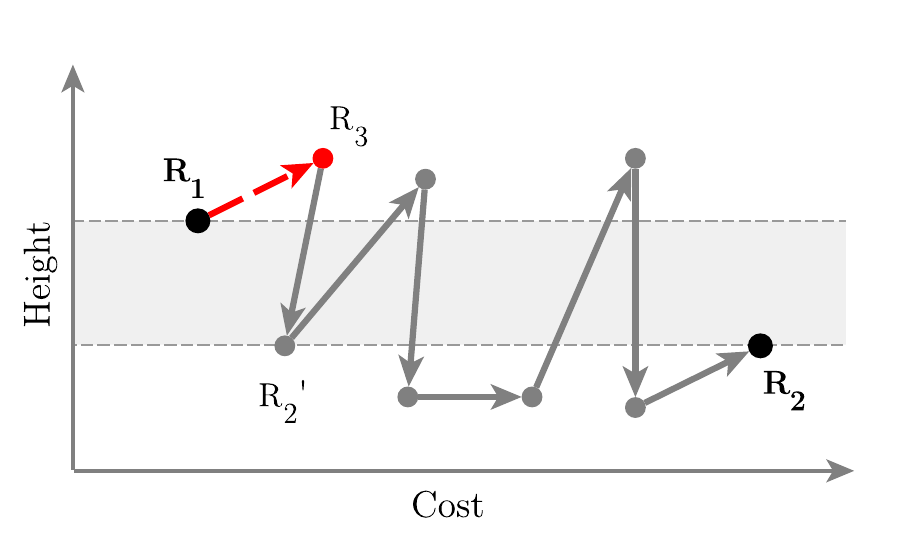}}
		\caption{Lemma \ref{lem:heightcost} base case.}
		\label{fig:thm3}
	\end{figure}

	We proceed by induction over length $n$ of the path from $R_1$ to $R_2$.
	Based on the previously-mentioned segment angle constraint, the smallest possible
	$n$ is 2.

	\paragraph{Base case $n = 2$.} This case is visualized by the path $R_1 \rightarrow R_3 \rightarrow R_2'$ in Figure \ref{fig:thm3}. If $R_3$ is below the shaded area, it cannot be
	to the right of $R_1$ due to the angle constraint, meaning $\cost{R_3} \leq \cost{R_1}$,
	and we are finished with the proof, since we have found an $R' = R_3$ such
	that $\cost{R'} \leq \cost{R_1}$ and $\expheight{R'} \leq \expheight{R_2'}$.

	Consider the case where $R_3$ is above the shaded area. In this case, the
	edge from $R_3$ to $R_2'$ represents a decrease in height. Examining only
	\textit{height-reducing}
	rewrite rules used in our algorithm (Section \ref{subsec:regex_rewrite}),
	we are limited to the following possibilities for $R_2',R_3$ (where $A$ is any
	regular expression). \\

    \begin{center}
    \scalebox{0.850}{
	\bgroup
	{
	\def\arraystretch{1.40}
	\begin{tabular}{|l||l|l||p{1.70in}|} \hline
		rule & $R_2'$ & $R_3$ & $R_1$ \\\hline\hline
		3 & $A$ & $A+0$ & $A+0$, $\boxed{A}+0$, $0+A$, $A$ \\\hline
		4 & $A$ & $A+A$ & $A+A,\boxed{A}+A$, $A+\boxed{A}$, $A$ \\\hline
		6 & $A$ & $1\cdot A$ & $1\cdot A,1\cdot \boxed{A}$, $A$ \\\hline
		7 & $A$ & $A\cdot 1$ & $A\cdot 1$, $\boxed{A}\cdot 1$, $A$ \\\hline
		10 & $0$ & $0\cdot A$ & $0\cdot A$, $0\cdot \boxed{A}$, $0$ \\\hline
		11 & $0$ & $A\cdot 0$ & $A\cdot 0$, $\boxed{A}\cdot 0$, $0$ \\\hline
		16 & $A^{\ast}$ & $1 + A{\cdot}A^{\ast}$ & $1 + A{\cdot}A^{\ast}$, $1 + \boxed{A}{\cdot}A^{\ast}$, $1 + A{\cdot}\boxed{A^{\ast}}$, $1 + A{\cdot}\boxed{A}{}^{\ast}$, $A{\cdot}A^{\ast}+1$, $A^{\ast}$\\\hline
		17 & $A^{\ast}$ & $1 + A^{\ast}{\cdot}A$ & $1 + A^{\ast}{\cdot}A$, $1 + A^{\ast}{\cdot}\boxed{A}$, $1 + \boxed{A^{\ast}}{\cdot}A$, $1 + \boxed{A}{}^{\ast}{\cdot}A$, $A^{\ast}{\cdot}A+1$, $A^{\ast}$ \\\hline
		18 & $A^{\ast}$ & $A^{\ast}{\cdot}A^{\ast}$ & $A^{\ast}{\cdot}A^{\ast}$, $\boxed{A^{\ast}}{\cdot}A^{\ast}$, $A^{\ast}{\cdot}\boxed{A^{\ast}}$, $\boxed{A}{}^{\ast}{\cdot}A^{\ast}$, $A^{\ast}{\cdot}\boxed{A}{}^{\ast}$, $A^{\ast}$ \\\hline
		19 & $A^{\ast}$ & $A^{\ast\ast}$ & $A^{\ast\ast}$, $\boxed{A^{\ast}}{}^{\ast}$, $\boxed{A}{}^{\ast\ast}$, $A^{\ast}$ \\\hline
		20 & $1$ & $1^{\ast}$ & $1^{\ast}$, $1$ \\\hline
		21 & $1$ & $0^{\ast}$ & $0^{\ast}$, $1$ \\\hline
	\end{tabular}
	}
	\egroup
    }
    \end{center}

	Based on these options for $R_3$, the corresponding options for $R_1$ are listed,
	using $\boxed{R}$ to denote a single rewrite rule applied to some subexpression of $R$,
	such that $\expheight{\boxed{R}} \leq \expheight{R}$.
	In all of these cases for $R_1$, we can find an $R'$ such that $\cost{R'}\leq\cost{R_1}$
	and $\expheight{R'}\leq\expheight{R_2'}$. In any case that does not contain $\boxed{A^{\ast}}$,
	we can simply let $R' = R_2'$. For cases that contain $\boxed{A^{\ast}}$, if $A$
	is of the form $R^{\ast}$ for some $R$, then we can let $R' = A$. Otherwise, the $\boxed{A^{\ast}}$ can instead be written as $\boxed{A}{}^{\ast}$, meaning we can again
	let $R' = R_2'$.

	\paragraph{Inductive step $n > 2$.} Assume the property holds for all $k < n$.
	Given the path of length $n$ between $R_1$ and $R_2$, consider the first segment
	$R_1 \rightarrow R_3$ (shown as the red/dashed arrow in Figure \ref{fig:thm3}).
	If $R_3$ is \textit{below} the shaded area, it cannot be to the right of $R_1$
	due to the angle constraint, meaning $\cost{R_3} \leq \cost{R_1}$ and
	$\expheight{R_3} \leq \expheight{R_2}$, so we are finished with the proof (we have found $R' = R_3$).
	
	Otherwise, if $R_3$ is \textit{above} the shaded area, we apply the induction
	hypothesis to $R_3, R_2$, giving us $R_4$ such that $\cost{R_4} \leq \cost{R_3}$
	and $\expheight{R_4} \leq \expheight{R_2}$.
	If $R_4$ is to the left of $R_1$, we are finished ($R' = R_4$).
	Otherwise, we can apply the induction hypothesis to $R_1,R_4$, giving us $R_5$
	such that $\cost{R_5} \leq \cost{R_1}$ and
	$\expheight{R_5} \leq \expheight{R_4}$. Since $\expheight{R_4} \leq \expheight{R_2}$, we
	have $\expheight{R_5} \leq \expheight{R_2}$, and we are finished (we found $R' = R_5$).
\end{proof}

\begin{customthm}{thm:heightcost}{Height vs. Cost}
\thmheightcost
\end{customthm}
\begin{proof}
We proceed by induction over the number of regular expressions with height between that
	of $R_1$ and $R_2$, i.e., $n = |\{R \;|\; \mathit{height}(R_1) > \mathit{height}(R) > \mathit{height}(R_2)\}|$.

	\paragraph{Base case $n = 0$.} This follows from Lemma \ref{lem:heightcost}.
	\paragraph{Inductive step $n > 0$.} Assume the Theorem holds for all $k < n$.
	Assume $\expheight{R_1} > \expheight{R_2}$ and $\cost{R_1} \leq \cost{R_2}$.
	Let $R_3$ be a regular expression such that $\expheight{R_1} > \expheight{R_3} > \expheight{R_2}$.

	Case I: $\cost{R_3} > \cost{R_1}$. Applying the induction hypothesis to $R_1,R_3$,
	we can obtain an $R_4$ such that $\expheight{R_4} \leq \expheight{R_3}$ and
	$\cost{R_4} \leq \cost{R_1}$.
	If $\expheight{R_4} \leq \expheight{R_2}$, we are finished with the proof (we found $R' = R_4$).
	Otherwise if $\expheight{R_4} > \expheight{R_2}$, since we have $\cost{R_4} \leq \cost{R_2}$
	from the induction hypothesis, we can apply the induction hypothesis to $R_4,R_2$ to
	obtain an $R_5$ such that $\expheight{R_5} \leq \expheight{R_2}$ and
	$\cost{R_5} \leq \cost{R_4}$. Since $\cost{R_4} \leq \cost{R_1}$, we are
	finished with the proof ($R' = R_5$).

	Case II: $\cost{R_3} \leq \cost{R_1}$.
	Using the induction hypothesis, we have $\cost{R_3} \leq \cost{R_2}$.
	Applying the induction hypothesis to $R_3, R_2$, we can obtain an $R_5$ such
	that $\expheight{R_5} \leq \expheight{R_2}$ and $\cost{R_5} \leq \cost{R_3}$.
	By Case II assumption, we have $\cost{R_5} \leq \cost{R_1}$,
	so we are finished with the proof ($R' = R_5$).
\end{proof}

\end{document}